\pdfoutput=1
\documentclass[a4paper,UKenglish,thm-restate,pdfa]{lipics-v2021-modified}

\usepackage{mathtools}					
\usepackage{xargs}                      
\usepackage[pdftex,dvipsnames]{xcolor}  
\usepackage{algorithm}
\usepackage{algpseudocode}
\usepackage[colorinlistoftodos,prependcaption,textsize=tiny]{todonotes}

\def\P{\ensuremath{\mathrm{P}}}

\def\NP{\ensuremath{\mathrm{NP}}}

\def\E{\ensuremath{\mathrm{E}}}

\def\NE{\ensuremath{\mathrm{NE}}}

\def\NEE{\ensuremath{\mathrm{NEE}}}
\def\FP{\ensuremath{\mathrm{FP}}}
\def\UP{\ensuremath{\mathrm{UP}}}
\def\DisjNP{\ensuremath{\mathrm{DisjNP}}}

\def\coNP{\ensuremath{\mathrm{coNP}}}

\def\coNE{\ensuremath{\mathrm{coNE}}}
\def\coNEE{\ensuremath{\mathrm{coNEE}}}

\def\TFNP{\ensuremath{\mathrm{TFNP}}}

\def\SPARSE{\ensuremath{\mathrm{SPARSE}}}

\def\TAUT{\ensuremath{\mathrm{TAUT}}}

\def\TFNP{\ensuremath{\mathrm{TFNP}}}

\def\PSPACE{\ensuremath{\mathrm{PSPACE}}}
\def\PH{\ensuremath{\mathrm{PH}}}
\def\NTIME{\ensuremath{\mathrm{NTIME}}}

\def\coNTIME{\ensuremath{\mathrm{coNTIME}}}

\def\N{\ensuremath{\mathrm{\mathbb{N}}}}

\def\leqmpp{\ensuremath{\leq_\mathrm{m}^\mathrm{pp}}}
\def\leqmp{\ensuremath{\leq_\mathrm{m}^\mathrm{p}}}
\def\psim{\ensuremath{\leq^\mathrm{p}}}

\DeclareMathOperator{\dom}{dom}
\DeclareMathOperator{\ran}{ran}

\DeclareMathOperator{\fp}{fp}

\def\sqsubsetneq{\mathrel{\sqsubseteq\kern-0.92em\raise-0.15em\hbox{\rotatebox{313}{\scalebox{1.1}[0.75]{\(\shortmid\)}}}\scalebox{0.3}[1]{\ }}}
\def\sqsupsetneq{\mathrel{\sqsupseteq\kern-0.92em\raise-0.15em\hbox{\rotatebox{313}{\scalebox{1.1}[0.75]{\(\shortmid\)}}}\scalebox{0.3}[1]{\ }}}
\newcommand{\card}[1]{\##1}
\def\runtime{\ensuremath{\mathrm{time}}}

\newcommand{\xTrueThenElse}[2]{\if\x1{#1}\else{#2}\fi}

\newcommandx{\unsure}[2][1=]{\todo[linecolor=red,backgroundcolor=red!25,bordercolor=red,#1]{#2}}
\newcommandx{\change}[2][1=]{\todo[linecolor=blue,backgroundcolor=blue!25,bordercolor=blue,#1]{#2}}
\newcommandx{\info}[2][1=]{\todo[linecolor=OliveGreen,backgroundcolor=OliveGreen!25,bordercolor=OliveGreen,#1]{#2}}
\newcommandx{\improvement}[2][1=]{\todo[linecolor=Plum,backgroundcolor=Plum!25,bordercolor=Plum,#1]{#2}}
\newcommandx{\thiswillnotshow}[2][1=]{\todo[disable,#1]{#2}}
\algdef{SE}[SUBALG]{Indent}{EndIndent}{}{\algorithmicend\ }%
\algtext*{Indent}
\algtext*{EndIndent}
\hideLIPIcs
\title{Recursive Jump Operators and Optimal Proof Systems}
\titlerunning{Recursive Jump Operators and Optimal Proof Systems}

\author{Fabian Egidy}{University of Würzburg, Germany}{fabian.egidy@uni-wuerzburg.de}{https://orcid.org/0000-0001-8370-9717}{supported by the German Academic Scholarship Foundation.}
\authorrunning{F.~Egidy}
\Copyright{Fabian Egidy}
\nolinenumbers
\acknowledgements{We wish to thank Christian Glaßer for his permanent advice and helpful feedback.}

\begin{document}
\maketitle
\begin{abstract}
    We study the relationship between the existence of optimal proof systems and recursive jump operators, two central open problems in proof complexity. For a set $L$, an optimal proof system is a strongest proof system in terms of proof length, whereas a recursive jump operator uniformly transforms any proof system for $L$ into a stronger one with respect to proof length, thereby witnessing non-optimality. It is clear that the existence of a recursive jump operator for $L$ rules out optimal proof systems for $L$. Khaniki (FOCS 2024) is interested in the converse of this implication and explicitly poses the following question, where $\TAUT$ denotes the set of propositional tautologies.
    \medskip
    \begin{itemize}
        \item[] Q:~~Does the non-existence of optimal proof systems for $\TAUT$ imply the existence of\\
        \phantom{Q:~~}recursive jump operators for $\TAUT$?
    \end{itemize}
    \medskip
    We generalize and address this question from both a relativized and an unrelativized perspective. We show that proving a positive answer for Q is provably hard by constructing the following oracle.
    \medskip
    \begin{itemize}
        \item[] $O$:~~The polynomial-time hierarchy is infinite, $\TAUT$ has no optimal proof systems, and\\
        \phantom{$O$:~~}$\TAUT$ has no recursive jump operators.
    \end{itemize}
    \medskip
    This shows that Khaniki's question can not be answered in the positive by relativizable means, even under the standard complexity-theoretic assumption that the polynomial-time hierarchy is infinite. 
    
    In contrast, we obtain positive results when the question Q is posed for sets different from $\TAUT$. We prove that the existence of recursive jump operators is upward closed under $\leqmp$-reducibility, a result that so far was only known for the non-existence of optimal proof systems.
    Furthermore, we show that the sets known to have no optimal proof systems by Messner (STACS 1999) in fact admit recursive jump operators. Thus, essentially all sets currently known to have no optimal proof systems have recursive jump operators.

\end{abstract}
\section{Introduction}
This paper studies the relationship of two longstanding open problems concerning Cook-Reckhow \cite{cr79} proof systems. A Cook-Reckhow proof system for a set $L$ is a polynomial time computable function $f$ whose range is $L$. Two such proof systems for the same set are compared with each other using the notion of simulation. A proof system $f$ simulates a proof system $g$ if there exists a polynomially-bounded total function $\pi$ such that $f(\pi(x)) = g(x)$ for all $x$. Intuitively, the proofs of $f$ are at most polynomially longer than the proofs of $g$. By Krajíček and Pudlák \cite{kp89}, a proof system $f$ for a set $L$ is optimal if $f$ simulates all proof systems for $L$. In other words, an optimal proof system for $L$ has at most polynomially longer proofs than any other proof system for $L$. 

One of the main open questions in proof complexity is whether $\TAUT$ has optimal proof systems, where $\TAUT$ denotes the set of propositional tautologies (cf.~\cite{kp89,kra19a}). In fact, it is even open whether there exists any set outside $\NP$ that has an optimal proof system (cf.~\cite{sad97,km98,mes00}).
While it is unclear which answer to this general question is to be expected \cite{hir10,bs11,kra19a}, with respect to $\TAUT$ a negative answer is conjectured \cite{kp89}.
\begin{spreadlines}{2ex}
\begin{align*}
&\text{\textbf{C1}: {\em $\TAUT$ has no optimal proof system.}}
\end{align*}
\end{spreadlines}

Another concept in proof complexity are recursive jump operators. A jump operator for a set $L$ is a function such that for any proof system $f$ for $L$, $J(f)$ is a proof system for $L$ that $f$ cannot simulate. If additionally, $J$ is recursive, we call $J$ a recursive jump operator for $L$\footnote{For improved clarity, we occasionally use the term ``non-recursive jump operator'' for a general jump operator.}. Intuitively, a recursive jump operator computes witnesses of non-optimality for given proof systems.

The first candidate jump operators for $\TAUT$ were proposed by Krajíček and Pudlák \cite{kp89} and implicitly already by Buss \cite{bus85}. Recently, Khaniki \cite{kha24} proposed a new candidate jump operator based on interactive proofs. Additionally, Khaniki is interested in the relationship between jump operators
and optimal proof systems.
Let us present a conjecture which was already implicitly contained in the work of Krajíček and Pudlák \cite{kp89} and explicitly formulated by Khaniki\footnote{Khaniki formulated the conjecture using partial recursive jump operators, but simultaneously shows their existence to be equivalent to the existence of recursive jump operators.} \cite{kha24}.
\begin{spreadlines}{2ex}
\begin{align*}
&\text{\textbf{C2}: {\em $\TAUT$ has a recursive jump operator.}}
\end{align*}
\end{spreadlines}
It is immediate that the existence of an optimal proof system for $\TAUT$ is equivalent to $\TAUT$ possessing no non-recursive jump operators \cite{kp89}. Similarly, it is clear that the existence of a recursive jump operator for $\TAUT$ implies that $\TAUT$ does not have optimal proof systems. So conjecture C2 implies conjecture C1. However, it is open whether the converse implication holds. This question is explicitly posed by Khaniki \cite[Problem 10]{kha24} and we state it as question Q1 below. We also generalize question Q1 to any set $L \notin \NP$\footnote{For all non-empty $L \in \NP$ it is clear that they have optimal proof systems and no jump operators.}.
\begin{spreadlines}{2ex}
\begin{align*}
&\text{\textbf{Q1}: {\em Does conjecture {\em C1} imply conjecture {\em C2}?}}
\\
&\text{\textbf{Q2}: {\em Does the non-existence of optimal proof systems for $L$ imply the existence of}}
\\[-2.2ex]
&\phantom{\text{\textbf{Q2}: }} \text{{\em recursive jump operators for $L$?}}
\end{align*}
\end{spreadlines}
Positive answers to these questions would show that non-optimality of proof systems always comes with a procedure computing witnesses of non-optimality, while a negative answer would separate existential and constructive notions of non-optimality in proof complexity.

In this paper we contribute to questions Q1 and Q2. 
We obtain hardness results for question Q1, but partial positive answers for question Q2. 
Before we make our contributions precise, we give an overview over Cook-Reckhow proof systems in complexity theory as well as the conjectures C1 and C2.

\subsection{Previous Work}

In the past 35 years optimal proof systems have been studied extensively and have shown to be connected to a wide range of research areas, including the separation of complexity classes \cite{cr79, kp89, km98, mes99, mes00, kra04b, kmt03, cfm14}, the existence of complete sets for promise classes \cite{raz94, sad97, kmt03, bkm09, pud17}, the existence of optimal acceptors \cite{kp89,sad99, mes99, mes00}, mathematical logic \cite{kp89, kra95, pud98, pud17}, descriptive complexity \cite{cf10a}, parameterized complexity \cite{cf10b}, learning theory \cite{ps22} and more (cf.~Krajíček's book on proof complexity \cite{kra19a}). We focus on the results from structural complexity theory that are most relevant to our work and refer the reader to the literature above for a comprehensive overview of optimal proof systems.

\subparagraph{(Non)-existence of optimal proof systems.} It is an easy fact that all sets in $\NP$ have optimal proof systems. Messner \cite{mes99} shows that there are sets without optimal proof systems. He proved that for any non-polynomial time-constructible function $t$ there are sets in $\coNTIME(t)$ that do not have optimal proof systems and that all $\leqmp$-hard sets for $\coNE$ do not have optimal proof systems. Surprisingly, this already covers all that we unconditionally know about the existence of optimal proof systems. Messner \cite{mes00} also shows that if $\NP = \coNP$, then arbitrary complex sets with optimal proof systems exist. In contrast, Egidy and Glaßer \cite{eg25} construct two oracles: relative to the first all sets outside non-deterministic quasi-polynomial time do not have optimal proof systems and relative to the second the polynomial-time hierarchy is infinite and all sets in $\PSPACE \setminus \NP$ do not have optimal proof systems. Hence, for most sets we do not know whether they have optimal proof systems and proving their existence would require to overcome at least the relativization barrier.

\subparagraph{Proof systems for TAUT.} The most interesting and best-studied proof systems are proof systems for $\TAUT$.
Krajíček and Pudlák \cite{kp89} showed that $\NE = \coNE$ implies the existence of optimal proof systems for $\TAUT$, which was improved to $\NEE = \coNEE$ by Köbler, Messner, and Torán \cite{kmt03}. Razborov \cite{raz94} discovered the first connection between optimal proof systems and promise classes by showing that the existence of optimal proof systems for $\TAUT$ implies the existence of $\leqmpp$-complete sets for $\DisjNP$, i.e., the class of disjoint $\NP$-pairs defined by Selman \cite{sel88} and Selman and Grollman \cite{gs88}. Initiated by this result, many further connections between proof systems for $\TAUT$ and disjoint $\NP$-pairs were obtained \cite{pud03, bey04, gssz04, bey06, bey07, gsz07, gsz09, bey09, bey10, bs11, ghsw23}. As already mentioned above, the relationship between proof systems and promise classes has turned out to be quite fruitful. These connections in addition to various connections of optimal proof systems to bounded arithmetic (cf.~\cite{kra95, pud98, kra19a}) lead into a research program initiated by Pudlák \cite{pud17}. Pudlák surveys on several conjectures concerning the existence of optimal proof systems for sets (including conjecture C1) and the existence of complete sets for promise classes (like $\UP$, $\NP \cap \coNP$, $\TFNP$) and asks for a systematic investigation of their relationships. In particular, Pudlák is interested in either proving further implications between conjectures or disproving them relative to an oracle. Since then several implications have been disproved relative to an oracle \cite{deg24, dg20, dos20b, dos20a, eeg22, gssz04, kha22}. 

\subparagraph{Optimal proof systems and recursive jump operators.} One approach to rule out the existence of optimal proof systems for $\TAUT$ is to provide a computable procedure that takes an arbitrary proof system $f$ and produces an improved system that efficiently proves tautologies for which $f$ admits only long proofs. Khaniki \cite{kha24} calls such procedures efficient jump operators. Several candidate efficient jump operators have been proposed \cite{kp89, bus85, kra04a, kha24} in the hope that they are useful to prove that even strong proof systems like extended Frege are not optimal. A similar concept are hard tautology generators, that informally compute families of tautologies that a given proof system requires super-polynomial size proofs on (cf.~\cite[Def.~3.2]{kha24} for a precise definition).

The existence of non-recursive jump operators for $\TAUT$ is equivalent to the non-existence of optimal proof systems for $\TAUT$ which is equivalent to the existence of hard tautology generators \cite{kp89, kha24}. Khaniki \cite{kha24} extends this relationship to the ``efficient'' setting, i.e., Khaniki shows that the existence of all of the following concepts are equivalent: partial recursive jump operators, recursive jump operators, polynomial time computable jump operators, partial recursive hard tautology generators, recursive hard tautology generators, polynomial time computable hard tautology generators. Khaniki \cite{kha24} explicitly asks whether the existence of recursive jump operators is also equivalent to the non-existence of optimal proof systems for $\TAUT$, to which our main result contributes to.

\subsection{Our Contribution}
We address the questions Q1 and Q2 from a relativized and an unrelativized perspective. 

\subparagraph{Contribution to Q1.}
Our main contribution to Q1 is the construction of the following oracle (cf.~Corollary \ref{cor:main-result}), which shows that proving a positive answer for Q1 is provably hard and answers an open question by Khaniki \cite{kha24}. 
\begin{spreadlines}{2ex}
\begin{align*}
&\text{$O$: } \text{ The polynomial-time hierarchy is infinite, $\TAUT$ has no optimal proof systems,}
\\[-2.2ex]
&\phantom{\text{$O$: }} \text{ and $\TAUT$ has no recursive jump operators}
\end{align*}
\end{spreadlines}
The oracle $O$ answers the question posed by Khaniki \cite[Problem 10]{kha24} in the sense that no relativizable proof can show the implication Q1. Even more, the hardness of Q1 remains true even if one assumes 
that the polynomial-time hierarchy is infinite. Note that from the point of view of answering Q1 in the negative, results like ours are the best we can hope for, because actually proving Q1 in the negative implies a proof of conjecture C1 and thus $\NP \neq \coNP$, which currently seems completely out of reach. Hence, a separation via oracles is a standard approach to separate two such longstanding complexity theoretic conjectures from each other, e.g., as reflected by Pudlák's \cite{pud17} approach to separate conjectures in his research program.

\subparagraph{Contribution to Q2.} In contrast to Q1, we obtain partial positive results regarding Q2, namely for sets currently known to lack optimal proof systems. Messner \cite{mes99} shows the existence of sets that do not have optimal proof systems. Together with the closure under $\leqmp$-reducibility for the class of all sets with optimal proof systems proved by Köbler and Messner \cite{km98}, Messner obtains that all $\leqmp$-hard sets for $\coNE$ do not have optimal proof systems. 

First, we show that the existence of recursive jump operators is upward closed under $\leqmp$-reducibility, i.e., the existence of recursive jump operators for $L$ implies the existence of recursive jump operators for any set $L'$ such that $L \leqmp L'$ (cf.~Theorem \ref{thm:jump-operator-reduction}). This is a dual result to the $\leqmp$-closure result of Köbler and Messner. Second, we analyze Messner's proof against the existence of optimal proof systems for sets. By doing so, we derive a recursive jump operator for these sets (cf.~Theorem \ref{thm:rec-jump-op-exist}). In total, our results show that any set shown to not have optimal proof systems by Messner \cite{mes99} in fact has recursive jump operators. This covers essentially all sets currently known to have no optimal proof systems.

\subsection{Sketch of the Oracle Construction}
We sketch the ideas behind the oracle construction. Our main goal is to obtain an oracle relative to which 
\begin{itemize}
    \item the polynomial-time hierarchy is infinite, 
    \item $\TAUT$ has no optimal proof systems, 
    \item $\TAUT$ has no recursive jump operators. 
\end{itemize}
We achieve this by constructing a sparse oracle relative to which there are no $\leqmpp$-complete disjoint $\NP$-pairs and no recursive jump operators for $\TAUT$. By a result of Razborov \cite{raz94}, our oracle then also has no optimal proof systems for $\TAUT$. Using the approach of Egidy and Glaßer \cite{eg25}, we can combine our sparse oracle with an oracle of Yao \cite{yao85} relative to which the polynomial-time hierarchy is infinite and obtain our desired oracle with combined properties.

We outline the construction of the sparse oracle by explaining how to achieve its two main properties separately. Combining both properties into a single oracle requires additional technical care and is achieved using a priority argument, which we defer to the full construction.

\subparagraph{Towards no recursive jump operators for TAUT.} Let $J_1, J_2, \dots$ and $P_1, P_2, \dots$ notations for the same standard enumeration of Turing transducers. For clarity, we use the enumeration $J_i$ when referring to candidate recursive jump operators and $P_i$ when referring to candidate proof systems. We treat all $J_i$ as candidate recursive jump operators for $\TAUT$ and diagonalize against them. Let $J_i$ be an arbitrary such candidate. 

The key idea is to construct a witness proof system $W_i$ that evaluates $J_i$ on its own program code and neutralizes $J_i$'s jump. We achieve this by using an effective version of Kleene's \cite{Kle38, Kle52, Soa87} fixed-point theorem. Informally, the Turing transducer $W_i$ works as follows:
\begin{itemize}
    \item $W_i$ computes $J_i(W_i) = b$, which is possible by an invocation of the fixed-point theorem.
    \item Depending on carefully encoded information in the oracle, $W_i$ may either:
    \begin{itemize}
        \item simulate the allegedly stronger proof system $P_b$, or
        \item fall back to a fixed, baseline proof system for $\TAUT$.
    \end{itemize}
\end{itemize}
There are three ways in which $J_i$ can fail to be a recursive jump operator. Namely, if there is some $a \in \N^+$ such that
\begin{enumerate}[F1]
    \item\label{enum:sketch-1} $J_i(a)$ does not halt.
    \item\label{enum:sketch-2} $J_i(a) = b$ and $P_a$ is a proof system for $\TAUT$ and $P_b$ is no proof system for $\TAUT$.
    \item\label{enum:sketch-3} $J_i(a) = b$ and $P_a$, $P_b$ are proof systems for $\TAUT$, but $P_a$ simulates $P_b$.
\end{enumerate}
Note that properties such as ``$P_b$ is a proof system for $\TAUT$'' are generally not robust under extensions of the oracle\footnote{Meaning, just because $P_b$ is a proof system for $\TAUT$ relative to some partial oracle $w$ does not mean that this holds relative to extensions of $w$.}. We have to account for that in the oracle construction, but for simplicity, we will ignore it in this sketch.

When diagonalizing against $J_i$, the oracle construction first attempts to realize Case F\ref{enum:sketch-1} or F\ref{enum:sketch-2} by an appropriate partial extension of the oracle.
If this is not possible, we enforce Case F\ref{enum:sketch-3} with $a \coloneqq W_i$. Since Case F\ref{enum:sketch-1} failed, $J_i(W_i) = b$ is defined. Since Case F\ref{enum:sketch-2} failed, $P_b$ is a proof system for $\TAUT$. We then add some suitable code word to the oracle that allows $W_i$ to simulate $P_b$. Then $J_i(W_i) = b$, $W_i$ and $P_b$ are proof systems for $\TAUT$ and $W_i$ simulates $P_b$. 

\subparagraph{Towards no complete disjoint sets.} 
We diagonalize against all disjoint $\NP$-pairs for being $\leqmpp$-complete. Given a pair of $\NP$-machines $(N_0, N_1)$, we define a corresponding witness disjoint pair $(A,B)$ whose elements are determined by the oracle on carefully chosen designated lengths. In particular, $0^n \in A$ if $n$ is of designated length and the oracle contains words from $00\Sigma^{n-2}$. Analogous for $B$ and $01\Sigma^{n-2}$. 

For each candidate reduction function $f \in \FP$, we ensure that $(A,B) \not \leqmpp (L(N_0),L(N_1))$ via $f$. We immediately diagonalize successfully, if we can add some word from $00\Sigma^{n-2}$ (resp., $01\Sigma^{n-2}$) to the oracle and $N_0(f(0^n))$ (resp., $N_1(f(0^n))$) rejects. If neither is possible, then both machines must accept on an exponential number of distinct oracle extensions obtained by words of length $n$. However, each leftmost accepting path queries only polynomially many oracle words. This allows us to extend the oracle by two words such that neither leftmost accepting path recognizes the second word. Consequently, $L(N_0)$ and $L(N_1)$ are not disjoint, and hence do not form a $\leqmpp$-complete disjoint $\NP$-pair.

\section{Basic Definitions and Notations.}
\subparagraph{Sets.} Let $\Sigma \coloneqq \{0,1\}$ be the default alphabet and $\Sigma^*$ be the set of finite words over $\Sigma$. The set of all (positive) natural numbers is denoted by $\N$ ($\N^+$). For $a,b \in \N$, we define $a\N+b \coloneqq \{a \cdot n + b \mid n \in \N\}$. We write the empty set as $\emptyset$. The cardinality of a set $A$ is denoted by $\card{A}$. For a set $A \subseteq \Sigma^*$ and a number $n \in \N$, we define $A^{\leq n} \coloneqq \{w \in A \mid |w| \leq n\}$ and analogous for $=$. For a clearer notation we use $\Sigma ^{\leq n}$ for ${\Sigma^*}^{\leq n}$ and $\Sigma^n$ for ${\Sigma^*}^{=n}$. The operators $\cup$, $\cap$, and $\setminus$ denote the union, intersection and set-difference. We denote the complement of a set $A \subseteq \Sigma^*$ by $\overline{A} = \{x \in \Sigma^* \mid x \notin A\}$. We may mix the definition of sets with regular expressions, e.g., the expression $0\Sigma^*$ denotes the set $\{0x \mid x \in \Sigma^*\}$ and the expression $010^*$ denotes the set $\{010^n \mid n \in \N\}$.

\subparagraph{Machines.} We use the default model of a Turing machine in both the deterministic and non-deterministic variant. The language decided by a Turing machine $M$ is denoted by $L(M)$. We use Turing transducer to compute functions. For a Turing transducer $F$ we write $F(x)=y$ when on input $x$ the transducer outputs $y$. Hence, a Turing transducer $F$ computes a function and we may denote ''the function computed by $F$'' by $F$ itself. For a Turing machine or Turing transducer $M$, we denote the number of steps the longest path of the computation $M(x)$ takes by $\runtime(M(x))$. When calling a non-deterministic Turing machine $M$ a non-deterministic Turing acceptor, then $\runtime(M(x))$ denotes the number of steps the shortest accepting path of the computation $M(x)$ takes, because Turing acceptors may endlessly loop on some computation paths. We denote the program code of a machine $M$ with respect to some machine enumeration by $\langle M \rangle$.

\subparagraph{Functions.}
The domain and range of a function $f$ are denoted by $\dom (f)$ and $\ran (f)$. We use the notion of time-constructible functions as described in \cite{ab09}. A function $f \colon \N \to \N$ is called non-polynomial if for every polynomial $p$ there is an $n \in \N$ such that $f(n) \geq p(n)$. The function computing the maximum element of a finite subset of $\N$ is denoted by $\max$. Also, let $\langle \cdot \rangle \colon \bigcup _{i \geq 0} \N^i \to (2\N+1)$ be an injective polynomial-time computable and polynomial-time invertible sequence encoding function such that $|\langle u_1, \dots , u_n \rangle | = 2(|u_1| + \cdots + |u_n| + n)+1$. 

\subparagraph*{Complexity classes.}
The definition of the basic complexity classes such as $\P$, $\FP$, $\NP$, $\coNP$, $\PH$, $\PSPACE$, $\E$, $\coNE$, $\NTIME$, $\coNTIME$ and the $\leqmp$-complete set $\TAUT$ for $\coNP$ can be found in \cite{ab09}. A set $S \subseteq \Sigma^*$ is sparse if there is a polynomial $p$ such that $\card{S^{\leq n}} \leq p(n)$ for all $n \in \N$. We denote the class of all sparse sets as $\SPARSE$. A disjoint $\NP$-pair is a pair $(A,B)$ of disjoint sets in $\NP$. Let $\DisjNP$ denote the class containing all disjoint $\NP$-pairs. 

We say that a set $A$ is polynomial-time many-one reducible to $B$, denoted by $A \leqmp B$, if there exists a function $f \in \FP$ such that $x \in A \Leftrightarrow f(x) \in B$ for all $x \in \Sigma^*$. We say that a disjoint $\NP$-pair $(A,B)$ is polynomial-time many-one reducible to $(C,D)$, denoted by $(A,B) \leqmpp (C,D)$, if there is a function $f \in \FP$ such that $f(A) \subseteq C$ and $f(B) \subseteq D$. For any class $\mathcal{C}$ and any notion of reducibility $\leq$ we say that $A$ is $\leq$-hard for $\mathcal{C}$ when $B \leq A$ for all $B \in \mathcal{C}$. If additionally $A \in \mathcal{C}$, we say that $A$ is $\leq$-complete for $\mathcal{C}$.

\subparagraph*{Proof systems.}
We use proof systems for sets defined by Cook and Reckhow \cite{cr79}. They define a function $f \in \FP$ to be a proof system for $\ran (f)$. Furthermore:
\begin{itemize}
\item A proof system $g$ is (p-)simulated by a proof system $f$, denoted by $g \leq f$ (resp., $g \psim f$), if there exists a total function $\pi$ (resp., $\pi \in \FP$) and a polynomial $p$ such that $|\pi(x)| \leq p(|x|)$ and $f(\pi(x)) = g(x)$ for all $x \in \Sigma^*$. In this context the function $\pi$ is called simulation function. Note that $g \psim f$ implies $g \leq f$.
\item A proof system $f$ is (p-)optimal for $\ran(f)$, if $g \leq f$ (resp., $g \psim f$) for all $g \in \FP$ with $\ran (g) = \ran (f)$.
\end{itemize}
In order to define jump operators, we must extend the definition of proof systems from polynomial-time computable functions to Turing transducers. There are two sensible options. First, one may consider any Turing transducer computing an $\FP$ function as a proof system, even if the transducer itself does not run in polynomial time. Alternatively, one may only consider Turing transducers running in polynomial time as proof systems. In this paper, we adopt the second option and thus identify proof systems with polynomial-time Turing transducers.

\subparagraph{Recursive jump operator.} 
A jump operator for a set $L$ is a function $J\colon \N^+ \to \N^+$ such that on input $a$, if $a$ is the code of a proof system $F$ for $L$, then $J(a)$ is the code of a proof system $G$ for $L$ such that $F$ cannot simulate $G$.\footnote{Note that the choice of which Turing transducers are regarded as a proof system is relevant here, since this determines the set of valid outputs as well as the set of inputs on which the jump operator is required to produce valid outputs.} If additionally $J$ is recursive, then $J$ is a recursive jump operator.

\section{Recursive Jump Operators for Sets Without Optimal Proof Systems}

Köbler and Messner \cite{km98} show that the class of sets with optimal proof systems is downward-closed under $\leqmp$-reducibility, with the trivial exception of $\emptyset$, which admits no proof system by definition. Our first result follows up on this by showing that the class of sets with recursive jump operators is upward-closed under $\leqmp$-reducibility. Second, we extract recursive jump operators from Messner's \cite{mes99} proof against the existence of optimal proof systems for sets in $\coNTIME(t)$ for non-polynomial and time-constructible $t$. Combining both results shows that the sets known to not have optimal proof system by Messner \cite{mes99} also have recursive jump operators. 

For the rest of this section, let $P_1, P_2, \dots$ be some standard enumeration of Turing transducers.
\begin{theorem}\label{thm:jump-operator-reduction}
   If $A \neq \emptyset$ has a recursive jump operator and $A \leqmp B$, then $B$ also has a recursive jump operator. 
\end{theorem}
\begin{proof}
    Let $J$ denote the recursive jump operator for $A$ and let $A \leqmp B$ via some function $f \in \FP$. Let $g$ be an arbitrary proof system for $B$ and let $\bot$ be an arbitrary element from $A$. Define
    \[g'(\langle x,w \rangle) \coloneqq 
    \begin{cases}
        x, & \text{if } g(w) = f(x)\\
        \bot, & \text{otherwise}
    \end{cases}.\]
    Observe that $g'$ is a proof system for $A$, because $x \in \ran(g')$ if and only if $f(x) \in \ran(g) = B$. Let $\langle g' \rangle$ denote the code of a Turing transducer computing $g'$. Let
    \[\langle h' \rangle \coloneqq J(\langle g' \rangle),\] 
    i.e., $h'$ is a proof system for $A$ such that $g'$ does not simulate $h'$. Finally, define 
    \[h(w) \coloneqq
    \begin{cases}
        f(h'(w')), & \text{if } w = 1w'\\
        g(w'),& \text{if } w = 0w'
    \end{cases}. 
    \] 
    Clearly $h$ is a proof system for $B$. We show that $g$ does not simulate $h$.

    Suppose for the sake of contradiction that $g$ simulates $h$ via $\pi$, i.e., $h(w) = g(\pi(w))$ for all $w \in \Sigma^*$. Then $g(\pi(1w)) = h(1w) = f(h'(w))$. Consequently, $g'(\langle h'(w), \pi(1w) \rangle) = h'(w)$, showing that any $h'$-proof has an at most polynomially longer $g'$-proof. This is a contradiction to $g'$ not simulating $h'$, so $g$ does not simulate $h$.
    
    Above construction provides a recursive jump operator $J'$ for $B$. Given any Turing transducer code $a$, $J'$ can compute the code $b$ of a Turing transducer working like $g'$ where $P_a$ is used for $g$. By assumption, $J'$ can use $J$ to compute $c \coloneqq J(b)$. Finally, $J'$ can compute the code $d$ of a Turing transducer computing $h$ where $P_a$ is used for $g$ and $P_c$ is used for $h'$. Hence, $J'$ is recursive. Also, observe that if $P_a$ is a proof system for $B$, then $J'$ can compute the codes $b$, $c$, and $d$ such that $P_b$, $P_c$, and $P_d$ are polynomial-time computable. Hence, $J'$ is recursive and if $P_a$ is a proof system for $B$, then $J'(a) = d$ where $P_d$ is a proof system for $B$ that is not simulated by $P_a$.      
\end{proof}
Above theorem can be helpful when approaching question Q1.
\begin{corollary}
    Let $L \neq \emptyset$ and $L \leqmp \TAUT$. If $L$ has recursive jump operators, then $\TAUT$ has recursive jump operators.
\end{corollary}
Next, we present the two most prominent results showing that there are sets without optimal proof systems. 
\begin{theorem}[\cite{mes99, mes00}]\label{thm:mes1}
    Let $t \colon \N \to \N$ be non-polynomial and time-constructible. Then there is a set $L \in \coNTIME(t)$ for which no optimal proof system exists.
\end{theorem} 
\begin{corollary}[\cite{mes99, mes00}]\label{cor:mes1}
   No set $\leqmp$-hard for $\coNE$ has an optimal proof system. 
\end{corollary}
We prove that the sets from Theorem \ref{thm:mes1} and Corollary \ref{cor:mes1} without optimal proof systems admit recursive jump operators. For this, we first define the sets considered in the proof of Theorem \ref{thm:mes1} and follow with our result.
\begin{definition}\label{def:Lt}
    Let $t \colon \N \to \N$ be non-polynomial and time-constructible. Let $N_1, N_2, \dots$ be a standard enumeration of non-deterministic Turing acceptors, and let $U$ be a universal non-deterministic acceptor that on input $0^i1x$ simulates $N_i(0^i1x)$ such that $\runtime(U(0^i1x)) \leq c_i \runtime(N_i(0^i1x)) + c_i$ for some constant $c_i$. For any $i \in \N^+$ define 
    \[L_{i,t} \coloneqq \{x \in 0^i10^* \mid U(x) \text{ does not accept in less than } t(|x|) \text{ steps}\}\]
    Define $L_t \coloneqq \bigcup _{i \in \N^+} L_{i,t}$.
\end{definition}
Note that by definition a time-constructible function $t$ must satisfy $t(n) \geq n$.
\begin{observation}\label{obs:Lt-in-coNTIME}
    It holds that $L_t \in \coNTIME(t)$.
\end{observation}
\begin{theorem}\label{thm:rec-jump-op-exist}
    Let $t \colon \N \to \N$ be non-polynomial and time-constructible. Then $L_t$ has a recursive jump operator.
\end{theorem}
\begin{proof}
    For $i \in \N^+$ let $D[i]$ denote a polynomial time Turing transducer with range $0^i10^*$ and $D[i](0^i10^j) = 0^i10^j$ for all $j \in \N$. For $a \in \N^+$ let $N[a]$ denote a non-deterministic Turing acceptor that on input $y$ computes $P_a(x) = y'$ for all $x$ and accepts on paths where $y=y'$. Then $N[a]$ accepts $L_t$ if and only if $\ran(P_a) = L_t$. It is clear that given $a \in \N^+$ one can compute $i \in \N^+$ such that $N_i = N[a]$. 

    Algorithm \ref{alg:3} defines how the jump operator $J$ for $L_t$ operates.
    \begin{algorithm}
    \caption{Jump operator $J$ for $L_t$}\label{alg:3}
    \begin{algorithmic}[1]
    \State \textbf{Input:} $a \in \N^+$\label{alg3:line-1}
    \State Compute code $i$ such that $N_i = N[a]$\label{alg3:line-2}
    \State Compute code $b$ such that $P_b(x) = \begin{cases}
        P_a(x') & \text{if }x=1x'\\
        D[i](x') & \text{if }x = 0x'
    \end{cases}$\label{alg3:line-3}
    \State \Return $b$\label{alg3:line-4}
    \end{algorithmic}
    \end{algorithm}

    It is clear that $J$ is recursive and that $b$ can be computed such that $P_b$'s runtime is in the order of magnitude of $P_a$ and $D[i]$. We show that $J$ is also a jump operator for $L_t$. Let $a \in \N^+$ such that $P_a$ is a proof system for $L_t$. Consequently, $N[a] = N_i$ accepts $L_t$. Then it holds that 
    \begin{align}\label{align:010-subset}
        0^i10^* \subseteq L_t,
    \end{align}
    because if $0^i10^j \notin L_t$ for some $j \in \N$, then $U(0^i10^j)$ accepts meaning that also $N_i(0^i10^j)$ accepts, a contradiction to $N_i$ accepting $L_t$. Furthermore, 
    \begin{align}\label{align:U-runtime}
        c_i \runtime(N_i(x)) + c_i \geq \runtime(U(x)) \geq t(|x|) \text{ for }x \in 0^i10^*,    
    \end{align}
    because if $U(x)$ accepts in less than $t(|x|)$ steps, then $x \notin L_t$, contradicting $0^i10^* \subseteq L_t$. 

    By (\ref{align:010-subset}), by $\ran(D[i]) = 0^i10^*$, and by $P_a$ being a proof system for $L_t$, it follows that $P_b$ is also a proof system for $L_t$. Furthermore, for $x \in 0^i10^*$, $P_b(0x) = D[i](x) = x$, so any $x \in 0^i10^*$ has a $P_b$-proof of length $|x|+1$. We show that the proof system $P_a$ can not simulate the proof system $P_b$, because of $P_b$'s short proofs for $0^i10^*$.
    
    Assume for the sake of contradiction that $P_a$ simulates $P_b$ via the function $\pi$ whose output-length is bounded by the polynomial $q$. Let $p$ be the polynomial bounding the runtime of $P_a$. For any $x \in L_t$, let $\hat{x}$ be the lexicographically shortest $P_a$-proof for $x$. Since $N_i(x)$ takes $|\hat{x}|$ steps to guess the shortest $P_a$-proof for $x$ and at most $p(|\hat{x}|)$ steps to compute $P_a(\hat{x})$, from (\ref{align:U-runtime}) we get
    \[c_i p(|\hat{x}|) + |\hat{x}| + c_i \geq c_i \runtime(N_i(x)) + c_i \geq t(|x|) \text{ for }x \in 0^i10^*.\]
    Since $P_a$ simulates $P_b$ via $\pi$, for all $x \in 0^i10^*$ it also holds that
    \[|\hat{x}| \leq q(|x|+1).\]
    Let $d_i \coloneqq max\{t(n) \mid n \leq i\}$. Then for all $n \in \N$ it holds that
    \[c_i p(q(n+1)) + q(n+1) + c_i + d_i \geq t(n).\]
    Note that $c_i p(q(n+1)) + q(n+1) + c_i + d_i$ is polynomial in $n$. This contradicts that $t$ is non-polynomial. So the assumption is false and $P_a$ does not simulate $P_b$.

    In total, $J$ is recursive and when $a$ is the code a of a proof system $P_a$ for $L_t$, then $b = J(a)$ is the code of a proof system $P_b$ for $L_t$ and $P_a$ does not simulate $P_b$. Hence, $J$ is a recursive jump operator for $L_t$.
\end{proof}
\begin{corollary}
    All sets $\leqmp$-hard for $\coNE$ have recursive jump operators.
\end{corollary}
\begin{proof}
    Let $t(n) \coloneqq 2^n$. By Theorem \ref{thm:rec-jump-op-exist} and Observation \ref{obs:Lt-in-coNTIME}, there is $L_t \in \coNTIME(t) \subseteq \coNE$ that has a recursive jump operator. Let $L$ be an arbitrary $\leqmp$-hard set for $\coNE$. Then $L_t \leqmp L$ and thus Theorem \ref{thm:jump-operator-reduction} gives that $L$ also has recursive jump operators.
\end{proof}

\section{Oracle Against a Positive Answer to Q1}
In this section, we construct a sparse oracle relative to which $\DisjNP$ has no $\leqmpp$-complete sets and $\TAUT$ has no recursive jump operator. By a result of Razborov \cite{raz94}, there are also no optimal proof systems for $\TAUT$ relative to the sparse oracle. Using the approach of Egidy and Glaßer \cite{eg25}, we can combine this sparse oracle with an oracle relative to which the polynomial-time hierarchy is infinite and obtain an oracle that has the combined properties. First, we introduce some further definitions and oracle-specific notations and follow with the rigorous construction of the oracle.

\subsection{Notation for the Oracle Construction}

We relativize the concept of Turing machines and Turing transducers by giving them access to a write-only oracle tape. We say that two oracle transducers $F$ and $G$ are equivalent, denoted by $F \equiv G$, if they compute the same function. We relativize machines, complexity classes, proof systems, jump operators, and (p-)simulation by defining them over machines with oracle access, i.e., whenever a Turing machine or Turing transducer is part of a definition, we replace them by an oracle Turing machine or an oracle Turing transducer. 
We indicate the access to some oracle $O$ in the superscript of the mentioned concepts, i.e., $\mathcal{C}^O$ for a complexity class $\mathcal{C}$ and $M^O$ for a Turing machine or Turing transducer $M$. We sometimes omit the oracles in the superscripts, e.g., when sketching ideas in order to convey intuition, but never in actual proofs. We also transfer all notations to the respective oracle concepts. Note that the $\coNP$-complete set $\TAUT$ can be generalized such that for every oracle $O$, $\TAUT^O$ is $\coNP^O$-complete. Fix some strictly monotone increasing polynomial $p_t$ such that $\overline{\TAUT^O}$ is decidable in non-deterministic time $p_t$ relative to all oracles $O$.

\subparagraph*{Words and sets.}
We denote the length of a word $w \in \Sigma^*$ by $|w|$. The empty word has length $0$ and is denoted by $\varepsilon$. The $(i+1)$-th letter of a word $w$ for $0 \leq i < |w|$ is denoted by $w(i)$, i.e., $w=w(0)w(1)\cdots w(|w|-1)$. If $v$ is a (strict) prefix of $w$, we write $v \sqsubseteq w$ ($v \sqsubsetneq w$) or $w \sqsupseteq v$ ($w \sqsupsetneq v$). 

We identify $\Sigma^*$ with $\N$ through a polynomial-time computable and polynomial-time invertible bijection $\mathrm{enc} \colon \Sigma^* \to \N$ defined as $\mathrm{enc}(w) = \sum _{i<|w|}(1+w(i))2^{i}$. Thus, we can treat words from $\Sigma^*$ as numbers from $\N$ and vice versa, which allows us to use notations, relations and operations of words for numbers and vice versa (e.g., we can define the length of a number by this). In particular, $\ran(\langle \cdot \rangle) \subseteq 0\Sigma^*$. Expressions like $0^i$ and $1^i$ are ambiguous, because they can be interpreted as words from $\Sigma ^*$ or numbers. We resolve this ambiguity by the context or explicitly naming the interpretation. A word $w \in \Sigma ^*$ can be interpreted as a set $\{i \in \N \mid w(i) = 1\}$
and subsequently as a partial oracle, which is defined for all words up to $|w|-1$. So $|w|$ is the first word that $w$ is not defined for, i.e., $|w| \in w1$ and $|w| \notin w0$. During the oracle construction we often use words from $\Sigma ^*$ to denote partially defined oracles. In particular, oracle queries for undefined words of a partial oracle are answered negatively.

For a deterministic Turing machine or transducer $F$ and an oracle $O$, we define $Q(F^O(x))$ as the set of words queried to the oracle by the computation $F^O(x)$. For a non-deterministic Turing machine $N$, we define $Q(N^O(x))$ as the set of words queried on the leftmost accepting computation path of $N^O(x)$. If $N^O(x)$ rejects, then $Q(N^O(x))$ is the empty set. Note that in the context of $Q$ we view a chained computation like $N^O(F^O(x))$ as one computation that first computes $F^O(x) = y$ and then $N^O(y)$.
\subparagraph*{Functions.}
A function $f'$ is an extension of a function $f$, denoted by $f \sqsubseteq f'$ or $f' \sqsupseteq f$, if $\dom(f) \subseteq \dom(f')$ and $f(x) = f'(x)$ for all $x \in \dom (f)$. If $x \notin \dom (f)$, then $f \cup \{x \mapsto y\}$ denotes the extension $f'$ of $f$ such that $f'(z) = f(z)$ for $z \not = x$ and $f'(x)=y$. We define polynomial functions $p_i \colon \N \to \N$ for $i \in \N^+$ by $p_i(n) \coloneqq n^i + i$.

\subparagraph{Definite computations.}
A computation $F^w(x)$ is called definite, if $w$ is defined for all words of length less than or equal to the runtime bound of $F^w(x)$.
A computation definitely accepts (resp., rejects), if it is definite and accepts (resp., rejects). We may combine computations to more complex expressions and call those definite, if every individual computation is definite, e.g., ``$N^w(F^w(x))$ is definite'' means $F^w(x)$ outputs some $y$ and is definite, and $ N^w(y)$ is definite.

\subparagraph{Enumerations.} For any oracle $O$ let $\{F_i^O\}_{i \in \N^+}$ be a standard enumeration of polynomial time oracle Turing transducers, where $F_i^O$ has running time exactly $p_i$. For any oracle $O$ let $\{N_i^O\}_{i \in \N^+}$ be a standard enumeration of non-deterministic polynomial time oracle Turing machines, where $N_i^O$ has running time exactly $p_i$. For any oracle $O$, let $\{J_i^O\}_{i \in \N^+}$ be a standard enumeration of oracle Turing transducers and let $\{P_i^O\}_{i \in \N^+}$ be a different notation for the same enumeration. For readability, we will use the machines $J_i$ when referring to candidate jump operators and the machines $P_i$ when referring to candidate proof systems, although $J_i$ is the same machine as $P_i$.

\subsection{Oracle Construction}
For this section, let $A \subseteq 1\Sigma^*$ (equivalently $A \subseteq 2\N$) be arbitrary. We will treat $A$ as some base oracle already present in the beginning of the construction. We will construct a sparse oracle $O \subseteq 0\Sigma^*$ (equivalently $O \subseteq 2\N+1$), such that the desired properties hold relative to $O \cup A$. 

\subparagraph{Witness pairs.}
For each candidate disjoint $\NP$-pair $(C,D)$ we define a witness pair $(A_m,B_m) \in \DisjNP$ showing that $(A_m,B_m) \not \leqmpp (C,D)$, i.e., $(A_m,B_m)$ is a witness that $(C,D)$ is no $\leqmpp$-complete pair for $\DisjNP$.

\begin{definition}[Stages]\label{def:Stages}
Define $e \colon \N \to \N$ such that $e(0) \coloneqq 2$ and $e(n) \coloneqq 2^{e(n-1)}$ for $n \in \N^+$.  For $m \in \N$ define $H_m \coloneqq \bigcup_{i\in \N}e(\langle m,i \rangle)$. Define $\mathcal{H} = \bigcup _{m \in \N} H_{m}$.
\end{definition}

\begin{observation}[Properties of $H_m$]\label{obs:Hi}
    $\mathcal{H} \in \P$ and $H_{m} \in \P$ for all $m \in \N$.
\end{observation}

\begin{definition}[Witness pair]\label{def:witness-sets}
For an oracle $O$ and $m \in \N$ define
\begin{align*}
A_{m}^O &\coloneqq \{0^n \mid n \in H_{m} \text{ and } O^{=n} \cap 00\Sigma^{n-2} \neq \emptyset\} \\
B_{m}^O &\coloneqq \{0^n \mid n \in H_{m} \text{ and } O^{=n}\cap 01\Sigma^{n-2} \neq \emptyset\}
\end{align*}
\end{definition}

\begin{observation}
    For all $m \in \N$ and all oracles $O$, it holds that $A_{m}^O, B_{m}^O \in \NP^O$.
\end{observation}

\begin{observation}\label{obs:witness-disjoint}
    Let $m \in \N$ and $O$ be an oracle. If for all $n \in H_{m}$ it holds that $O^{=n} \cap 00\Sigma^{n-2} = \emptyset$ or $O^{=n} \cap 01\Sigma^{n-2} = \emptyset$, then $(A_{m}^O, B_{m}^O) \in \DisjNP^O$.
\end{observation}

\subparagraph{Witness proof systems.} 
For each candidate recursive jump operator $J_i$ we define many witness proof systems $W_{i,m}$ possibly witnessing that $J_i$ is no recursive jump operator for $\TAUT$, i.e., $J_i(\langle W_{i,m} \rangle) = b$ and either $P_b$ is no proof system for $\TAUT$ or $W_{i,m}$ simulates $P_b$.

\begin{theorem}[Relativized effective version of Kleene's fixed-point theorem \cite{Kle38,Kle52,Soa87}]\label{thm:fixed-point}
    There is a recursive function $\fp \colon \N^+ \to \N^+$ such that for all $i \in \N^+$ it holds that if $J_i$ is total, then $P_{\fp(i)}^O \equiv P_{J_i(\fp(i))}^O$ relative to any oracle $O$.
\end{theorem}
We use this fixed-point theorem to define our witness proof systems $W_{i,m}$ such that they evaluate $J_i$ on their own program code. For this, let $i,m \in \N^+$ and let $s_{i,m} \colon \N^+ \to \N^+$ be a function that returns the code of the oracle-program defined in Algorithm \ref{alg:1} when given $z \in \N^+$ as input. Note that $O$ is a placeholder for an arbitrary oracle, let $F$ be a proof system for $\TAUT$ relative to all oracles, and recall that $p_t$ is the polynomial bounding the number of steps to non-deterministically decide $\overline{\TAUT}$ relative to any oracle. We provide some intuition of Algorithm \ref{alg:1} after Observation \ref{obs:Wim}. 

\begin{algorithm}
\caption{Witness proof system with code $s_{i,m}(z)$}\label{alg:1}
\begin{algorithmic}[1]
\State \textbf{Input:} $\langle x, y \rangle \in \Sigma^*$\label{alg1:line-1}
\State Compute $b \coloneqq J_i^O(z)$ \Comment{obtain ``jumped'' proof system}\label{alg1:line-2}
\If{$\langle 0^i,0^m,0^n \rangle \in O$ for any $n < |y|$} \Comment{oracle permits simulation}\label{alg1:line-3}
    \If{$\langle 1^i,1^m,1^n \rangle \notin O$ for all $n \leq p_t(|y|)$} \Comment{permission not revoked}\label{alg1:line-4}
        \If{$P_b^O(x)$ halts after $\leq |y|$ steps}\Comment{ensures polytime simulation}\label{alg1:line-5}
            \State \Return $P_b^O(x)$ \Comment{simulate ``jumped'' proof system}\label{alg1:line-6}
        \Else\label{alg1:line-7}
            \State \Return $F^O(x)$\label{alg1:line-8}
        \EndIf\label{alg1:line-9}
    \EndIf\label{alg1:line-10}
\EndIf\label{alg1:-line-11}
\State \Return $F^O(x)$\label{alg1:line-12}
\end{algorithmic}
\end{algorithm}
For all $i,m \in \N^+$, let $\langle s_{i,m} \rangle$ denote the program code of a polynomial-time Turing transducer computing $s_{i,m}$. The following observation shows that such a Turing transducer exists.
\begin{observation}\label{obs:fixed-point} Let $O$ be an arbitrary oracle. For all $i,m \in \N^+$ it holds that
    \begin{romanenumerate}
        \item $s_{i,m} \in \FP$.\label{obs:fixed-point-i}
        \item $P_{\fp(\langle s_{i,m} \rangle)}^O \equiv P_{s_{i,m}(\fp(\langle s_{i,m} \rangle))}^O$.\label{obs:fixed-point-iii} 
    \end{romanenumerate}
\end{observation}
\begin{proof}
    Let $i,m \in \N^+$ be arbitrary. It is clear that (\ref{obs:fixed-point-i}) holds. 
    Property (\ref{obs:fixed-point-iii}) follows by the totality of $s_{i,m}$ and Theorem \ref{thm:fixed-point} invoked with $\langle s_{i,m} \rangle$ for $i$.
\end{proof}
\begin{definition}[Witness proof systems]\label{def:witness-function-jump-operator}
    Let $O$ be an arbitrary oracle and let $i,m \in \N^+$. Define $W_{i,m}^O \coloneqq P_{\fp(\langle s_{i,m} \rangle)}^O \equiv P_{s_{i,m}(\fp(\langle s_{i,m} \rangle))}^O$.   
\end{definition}
\begin{observation}\label{obs:Wim}
    The witness proof system $W_{i,m}$ has code $\fp(\langle s_{i,m} \rangle)$ and works like Algorithm \ref{alg:1} for $i$, $m$, and $z \coloneqq \fp(\langle s_{i,m} \rangle)$.
\end{observation}
With this, for every candidate recursive jump operator $J_i$, we have infinitely many witness proof systems $W_{i,m}$ behaving like Algorithm \ref{alg:1} where $z = \langle W_{i,m} \rangle$. So $W_{i,m}$ first evaluates $J_i$ on its own program code and obtains the code of a supposedly stronger proof system $P_b$. Depending on information in the oracle, $W_{i,m}$ is allowed to simulate $P_b$. Here, the oracle should only allow the simulation of $P_b$ if $\ran(P_b) \subseteq \TAUT$. Lastly, $W_{i,m}$ only simulates $P_b$ if the simulation can be done in a polynomial number of steps, ensuring that $W_{i,m}$ remains polynomial-time computable.

The following results show that if $P_b$ is a proof system and the oracle has the right conditions, $W_{i,m}$ is a proof system for $\TAUT$ that simulates $P_b$.
\begin{lemma}\label{lem:Wim-polytime}
    For all $i,m \in \N^+$ and all oracles $O$ it holds that if $b \coloneqq J_i^O(\fp(\langle s_{i,m} \rangle))$ is defined, then $W_{i,m}^O$ is total and runs in polynomial time. If additionally $\ran(P_b^O) \subseteq \TAUT^O$, then $W_{i,m}^O$ is a proof system for $\TAUT^O$.
\end{lemma} 
\begin{proof}
    Consider $W_{i,m}^O$. Line \ref{alg1:line-2} does not depend on the size of the input and halts by assumption, thus requiring only constant time. Lines \ref{alg1:line-3} and \ref{alg1:line-4} are about at most polynomially many polysize oracle queries. Lines \ref{alg1:line-5} and \ref{alg1:line-6} are about simulating computations for a polynomially bounded number of steps. Finally, lines \ref{alg1:line-8} and \ref{alg1:line-12} simulate a proof system for $\TAUT^O$. Hence, $W_{i,m}^O$ runs in polynomial time.

    Line \ref{alg1:line-2} computes the number $b$ from the assumption. Any output is either computed by $P_b^O$ or $F^O$, whose ranges are subsets of $\TAUT^O$. Hence, $\ran(W_{i,m}) \subseteq \TAUT^O$. Furthermore, line \ref{alg1:line-3} never evaluates to TRUE for inputs of the form $\langle x, \varepsilon \rangle$, so line \ref{alg1:line-12} is executed for all $x \in \Sigma^*$ at least once, resulting in $\ran(W_{i,m}^O) \supseteq \ran(F^O) = \TAUT^O$. Consequently, $\ran(W_{i,m}^O) = \TAUT^O$. Together with $W_{i,m}^O$ running in polynomial time, we get that $W_{i,m}^O$ is a proof system for $\TAUT^O$.
\end{proof}
\begin{lemma}\label{lem:Wim-simulate}
    Let $i,m \in \N^+$ and $O$ be an arbitrary oracle. Let $b \coloneqq J_i^O(\fp(\langle s_{i,m} \rangle))$ be defined and let $P_b^O$ be a proof system for $\TAUT^O$. If $\langle 0^i, 0^m, 0^n \rangle \in O$ for some $n \in \N$ and $\langle 1^i, 1^m, 1^n \rangle \notin O$ for all $n \in \N$, then it holds that $W_{i,m}^O$ is a proof system for $\TAUT^O$ simulating $P_b^O$.
\end{lemma}
\begin{proof}
    By Lemma \ref{lem:Wim-polytime} and $P_b^O$ being a proof system for $\TAUT^O$, it follows that $W_{i,m}^O$ is also a proof system for $\TAUT^O$. 

    Let $p$ be a strictly monotone increasing polynomial bounding the runtime of the proof system $P_b^O$. Let $x$ be arbitrary such that line \ref{alg1:line-3} evaluates to TRUE on input $\langle x, 0^{p(|x|)} \rangle$. Only finitely many $x$ do not satisfy this.
    Also by assumption, line \ref{alg1:line-4} always evaluates to TRUE. Then for the input $\langle x, 0^{p(|x|)} \rangle$ line \ref{alg1:line-5} evaluates to TRUE, because $P_b^O(x)$ halts in $\leq p(|x|)$ steps. Hence, for all but finitely many $x$ it holds that $W_{i,m}(\langle x, 0^{p(|x|)} \rangle) = P_b^O(x)$. Consequently, $W_{i,m}^O$ has at most polynomially longer proofs than $P_b^O$, so $W_{i,m}^O$ simulates $P_b^O$.
\end{proof}

\subparagraph{Valid oracles.}
During the construction of the oracle, we successively add requirements that we maintain. These are specified by a function $r \colon \N^+ \times \N \times \N \mapsto \N$, called \emph{requirement function}. Recall that $A \subseteq 1\Sigma^*$ denotes an arbitrary base oracle. An oracle $w$ is $r$-valid for a requirement function $r$, if it satisfies the following requirements (let $i,j,m \in \N^+$):
\begin{enumerate}[V1]
    \item\label{V1} $w \subseteq 0\Sigma^*$ and  $\card{(w \cap \Sigma^{n})} \leq 2$ for all $n \in \N^+$.
    
    (Meaning: $w \subseteq 2\N+1$ and $w$ is sparse.)

    \item\label{V2} If $\langle 0^i, 0^m, 0^n \rangle \in w$ or $\langle 1^i, 1^m, 1^n \rangle \in w$ for some $n \in \N$, then $n > \runtime(J_i^{w \cup A}(\fp(\langle s_{i,m} \rangle)))$.
    
    (Meaning: Codings for $W_{i,m}$ can not be queried during the simulation in line \ref{alg1:line-2}.)
    \item\label{V3} If $r(i,0,0) = m > 1$, then $\langle 1^i, 1^m, 1^n \rangle \notin w$ for all $n \in \N$.
    
    (Meaning: Here, $W_{i,m}$ commits to also return by line \ref{alg1:line-6}.)

    \item\label{V4} If $r(i,j,0) = m > 0$, then $\card{(w \cap \Sigma^{n})} \leq 1$ for all $n \in H_m$.
    
    (Meaning: $(A_m,B_m)$ is disjoint relative to $A$ combined with any extension of $w$.)

\end{enumerate}
We will prove that V\ref{V1}, V\ref{V2}, V\ref{V3} and V\ref{V4} are satisfied for various oracles and requirement functions. To prevent confusion, we use the notation V\ref{V1}($w$), V\ref{V2}($w$), V\ref{V3}($w,r$) and V\ref{V4}($w,r$) to clearly state the referred oracle $w$ and requirement function $r$.
\begin{observation}\label{obs:remain-valid}
    If $w$ is an $r$-valid oracle, then $w0$ is $r$-valid.
\end{observation}
\begin{observation}\label{obs:valid-monotone}
    If $w$ is $r$-valid, then $w$ is also $r'$-valid for any $r' \sqsubseteq r$.
\end{observation}
\begin{observation}\label{obs:valid-monoton2}
    If $w$ is $r$-valid and $v \sqsubseteq w$, then $v$ is also $r$-valid.
\end{observation}
\subparagraph{Oracle construction.}
The oracle construction will take care of the following set of tasks $\{\tau_{i,j,k} \mid i \in \N^+, j,k \in \N\}$ that are defined below. Let $\mathcal{T}$ be an enumeration of these tasks with the property that $\tau_{i,j,k}$ appears earlier than $\tau_{i,j,k+1}$. In each step we treat the smallest task in the order specified by $\mathcal{T}$, and after treating a task we remove it from $\mathcal{T}$.

In step $s = 0$ we define $w_0 \coloneqq \varepsilon$ and $r_0$ as the nowhere defined requirement function. In step $s > 0$ we define $w_s \sqsupsetneq w_{s-1}$ and $r_s \sqsupsetneq r_{s-1}$ such that $w_s$ is $r_s$-valid by treating the earliest task $\tau$ in $\mathcal{T}$ and removing it from $\mathcal{T}$. We do this according to the following procedure (let $i,j,k \in \N^+$):
\bigskip
\\
\centerline{\textbf{Task} $\tau_{i,0,0}$}
\begin{enumerate}
\item\label{task1:1} If there exists an $r_{s-1}$-valid partial oracle $v \sqsupsetneq w_{s-1}$ such that there are $a \in \N^+$ and $x \in \Sigma^*$ with 
\begin{enumerate}[a.]
    \item\label{task1:1-i} $P_a$ is a proof system for $\TAUT$ relative to $A$ combined with any fully defined $r_{s-1}$-valid extension of $v$,
    \item\label{task1:1-ii} $J_i^{v \cup A}(a) = b$ is definite,
    \item\label{task1:1-iii} $P_b^{v \cup A}(x) \notin \TAUT^{v \cup A}$ is definite,
\end{enumerate} 
then define $w_s \coloneqq v$ and $r_s \coloneqq r_{s-1} \cup \{(i,0,0) \mapsto 0\}$. 

(Meaning: $J_i$ will not be a jump operator, because $J_i$ maps some code $a$ to some code $b$, while $P_a$ will be and $P_b$ will not be a proof system for $\TAUT$ relative to the final oracle.)
\item\label{task1:2} Otherwise, if there exists 
an $a \in \N^+$ such that $J_i(a)$ never halts relative to $A$ combined with any $r_{s-1}$-valid partial extension of $w_{s-1}$, then define $w_s \coloneqq w_{s-1}0$ and $r_s \coloneqq r_{s-1} \cup \{(i,0,0) \mapsto 1\}$.

(Meaning: $J_i$ will not be a recursive jump operator, because $J_i$ will not be total.)

\item\label{task1:3} Otherwise,
\begin{enumerate}[a.]
    \item\label{task1:3-i} let $m \coloneqq |w_{s-1}|+2$,
    \item\label{task1:3-ii} let $v' \sqsupseteq w_{s-1}$ be the smallest $r_{s-1}$-valid partial oracle such that $J_i^{v' \cup A}(\fp(\langle s_{i,m} \rangle))$ halts in less than $||v'||$ steps,
    \item\label{task1:3-iii} and let $v$ be an extension of $v'$ via zeros such that $\langle 0^i, 0^m, 0^{||v'||} \rangle \in v1$.
\end{enumerate}
Then define $w_s \coloneqq v1$ and $r_s \coloneqq r_{s-1} \cup \{(i,0,0) \mapsto m\}$.

(Meaning: Choose a sufficiently large $m$ such that $w_{s-1}$ can not contain words of the form $\langle 0^i, 0^m, \cdot \rangle$ and let $w_s$ contain the permission such that $W_{i,m}$ can return by line \ref{alg1:line-6}, thereby ruling out $J_i$ as a jump operator. Statement \ref{task1:3-ii} is possible since Case \ref{task1:2} failed.)
\end{enumerate}
\ 
\\   
\centerline{\textbf{Task} $\tau_{i,j,0}$}
\begin{enumerate}
    \item\label{task2:1} If there exists an $r_{s-1}$-valid partial oracle $v \sqsupsetneq w_{s-1}$ such that there is an $x$ with $N_i^{v \cup A}(x)$ and $N_j^{v \cup A}(x)$ accept definitely, then define $w_s \coloneqq v$ and $r_s \coloneqq r_{s-1} \cup \{(i,j,0) \mapsto 0\}$. Remove all tasks $\tau_{i,j,k}$ with $k \in \N^+$ from $\mathcal{T}$.
     
    (Meaning: The pair $(L(N_i),L(N_j))$ is not disjoint.)
    \item\label{task2:2} Otherwise, define $w_s \coloneqq w_{s-1}0$ and $r_s \coloneqq r_{s-1} \cup \{(i,j,0) \mapsto |w_s|+2\}$.
    
    (Meaning: Assign the pair $(L(N_i),L(N_j))$ some witness pair $(A_{|w_s|+2},B_{|w_s|+2})$ and promise to keep the witness-pair disjoint.)
\end{enumerate}
\ 
\\
\centerline{\textbf{Task} $\tau_{i,j,k}$}
\begin{enumerate}
    \item\label{task3:1} Let $m \coloneqq r_{s-1}(i,j,0)$ and let $v \sqsupseteq w_{s-1}$ be $r_{s-1}$-valid such that $||v|| = 0^n$ with $n \in H_m$ and $2^{n-2} > 6(p_{i+j}(p_k(n)))^2+1$. Choose 
\begin{enumerate}[a.]
    \item\label{task3:1-i} either $z \in 00\Sigma^{n-2}$ such that $N_i^{v \cup \{z\} \cup A}(F_k^{v \cup \{z\} \cup A}(0^n))$ rejects
    \item\label{task3:1-ii} or $z \in 01\Sigma^{n-2}$ such that $N_j^{v \cup \{z\} \cup A}(F_k^{v \cup \{z\} \cup A}(0^n))$ rejects.
\end{enumerate}
Then define $r_s \coloneqq r_{s-1} \cup \{(i,j,k) \mapsto 0\}$ and define $w_s \sqsupseteq v$ such that it contains the same words as $v \cup \{z\}$ and $||w_s|| = 2p_{i+j}(p_k(n))+1$. 

(Meaning: The witness pair $(A_m,B_m)$ is not reducible to $(L(N_i),L(N_j))$ via $F_k$. It must be shown that $z$ can be chosen as stated. Intuitively, this follows from $r_{s-1}(i,j,0) \neq 0$.)
\end{enumerate}

\begin{definition}[Desired Oracle]\label{def:desired-oracle}
    Define $O \coloneqq \bigcup_{s \in \N} w_s$ and $r \coloneqq \bigcup _{s \in \N} r_s$.
\end{definition}
It remains to show that Definition \ref{def:desired-oracle} is well-defined (i.e., that all steps of the oracle construction can be performed and makes real progress) and that $O$ satisfies the desired properties. In Lemma~\ref{lem:ws-rs-valid}, we show that $O$ can be constructed as stated, in particular, we show that Case \ref{task1:3} of task $\tau_{i,0,0}$ and Case \ref{task3:1} of task $\tau_{i,j,k}$ are possible. The Propositions \ref{prop:O-is-sparse}, \ref{prop:no-jump-operator}, and \ref{prop:no-optimal-ps} show that $O$ satisfies all the desired properties.

\begin{lemma}\label{lem:ws-rs-valid}
    For all $s \in \N$ it holds that $w_s$, $r_s$ are well-defined, $w_s \sqsubsetneq w_{s+1}$, $r_s \sqsubsetneq r_{s+1}$ and that $w_s$ is $r_s$-valid.
\end{lemma}
\begin{proof}
    It is clear that $w_0$ and $r_0$ are well-defined and that $w_0$ is $r_0$-valid. Let $s \in \N^+$ and $w_{s'}$ be $r_{s'}$-valid for $s' < s$. The following three claims show that $w_s$ is both well-defined and $r_s$-valid when defined by any given task. Let $i,j,k \in \N^+$.
    \begin{claim}
        $w_s$ is well-defined and $r_s$-valid when defined by task $\tau_{i,0,0}$.
    \end{claim}
    \begin{claimproof}
        If $w_s$ is defined by Case \ref{task1:1} of task $\tau_{i,0,0}$, then $w_s$ is explicitly chosen as an $r_{s-1}$-valid extension of $w_{s-1}$. If $w_s$ is defined by Case \ref{task1:2} of task $\tau_{i,0,0}$, then $w_{s} = w_{s-1}0$. By Observation \ref{obs:remain-valid}, it holds that $w_s$ is $r_{s-1}$-valid. Furthermore, $w_s$ is $r_s$-valid in both cases, because no additional requirement must hold by the extension of $r_{s-1}$ to $r_s$ via $\{(i,0,0) \mapsto 0\}$ or via $\{(i,0,0) \mapsto 1\}$.
    
        Let $w_s$ be defined by Case \ref{task1:3} of task $\tau_{i,0,0}$. Let $m \coloneqq |w_{s-1}|+2$ according to Statement \ref{task1:3-i}. First, we show that $w_s$ is well-defined, i.e., that $v'$ from Statement \ref{task1:3-ii} can be chosen as stated. Observe that since Case \ref{task1:2} failed, for every $a \in \N^+$ there is some $r_{s-1}$-valid partial oracle $v' \sqsupseteq w_{s-1}$ such that $J_i^{v' \cup A}(a)$ halts. Together with Observation \ref{obs:remain-valid}, there must be some smallest $r_{s-1}$-valid partial oracle $v' \sqsupseteq w_{s-1}$ such that $J_i^{v' \cup A}(\fp(\langle s_{i,m} \rangle))$ halts in less than $||v'||$ steps.

        Next, we show that $w_s$ is $r_{s-1}$-valid. By choice, we have that $v'$ is $r_{s-1}$-valid. By Observation \ref{obs:remain-valid}, $v$ from Statement \ref{task1:3-iii} is $r_{s-1}$-valid and it holds that $|v| = \langle 0^i, 0^m, 0^{||v'||} \rangle$. Case \ref{task1:3} defines $w_s \coloneqq v1$. V\ref{V1}($v1$) and V\ref{V4}($v1,r_{s-1}$) are satisfied, because $v \cap \Sigma^{||v||} = \emptyset$ due to the extension of $v'$ to $v$ via zeros (note that $||v|| > ||v'||$). V\ref{V2}($v1$) is satisfied, because by Statement \ref{task1:3-ii}, 
        \[||v'|| > \runtime(J_i^{v' \cup A}(\fp(\langle s_{i,m} \rangle))),\]
        and thus also 
        \[||v'|| > \runtime(J_i^{v1 \cup A}(\fp(\langle s_{i,m} \rangle))).\]
        V\ref{V3}($v1,r_{s-1}$) is satisfied, because $\langle 0^i, 0^m, 0^{||v'||} \rangle$ is not of the type $\langle 1^+, 1^+, 1^* \rangle$. Hence, $w_s = v1$ is $r_{s-1}$-valid.

        Finally, we show that $w_s = v1$ is even $r_s$-valid. Note that only V\ref{V3} is affected by the extension of $r_{s-1}$ to $r_s$ via $\{(i,0,0) \mapsto m\}$. Observe that 
        \[||v'|| \leq \max(m,\runtime(J_i^{v' \cup A}(\fp(\langle s_{i,m} \rangle)))+1).\] 
        Hence, either by the size of $m$ or by V\ref{V2}($v1$), $\langle 1^i, 1^m, 1^n \rangle \notin v'$ for all $n \in \N$. Since $v$ is an extension of $v'$ via zeros, also $\langle 1^i, 1^m, 1^n \rangle \notin v$ for all $n \in \N$. Since $|v| = \langle 0^i, 0^m, 0^{||v'||} \rangle$, $\langle 1^i,1^m,1^n \rangle \notin v1$ for all $n \in \N$. Consequently, V\ref{V3}($v1,r_s$) is satisfied and $w_s$ is $r_s$-valid.
        
        In total, in any case, $w_s$ is well-defined and $r_s$-valid.
    \end{claimproof}
    \begin{claim}
        $w_s$ is well-defined and $r_s$-valid when defined by task $\tau_{i,j,0}$.
    \end{claim}
    \begin{claimproof}
        If $w_s$ is defined by Case \ref{task2:1} of task $\tau_{i,j,0}$ then $w_s$ is explicitly chosen as an $r_{s-1}$-valid extension of $w_{s-1}$. Furthermore, no additional requirement must hold by the extension of $r_{s-1}$ to $r_s$ via $\{(i,j,0) \mapsto 0\}$. So $w_s$ is also $r_s$-valid.

        Otherwise $w_s = w_{s-1}0$ and $m \coloneqq r_s(i,j,0) > 0$. By Observation \ref{obs:remain-valid}, $w_s$ remains $r_{s-1}$-valid. The extension of $r_{s-1}$ to $r_s$ via $\{(i,j,0) \mapsto |w_s|+2\}$ affects V\ref{V4}. Since $||w_s|| < m$, $w_s$ is not defined for words of any length in $H_m$. Hence, V\ref{V4}($w_s,r_s$) also remains satisfied. Consequently, $w_s$ is $r_s$-valid also in this case. 
    \end{claimproof}
    \begin{claim}
        $w_s$ is well-defined and $r_s$-valid when defined by task $\tau_{i,j,k}$.
    \end{claim}
    \begin{claimproof}
        The main challenge is to show that processing this task is possible. By the ordering of $\mathcal{T}$, $\tau_{i,j,0}$ is treated before $\tau_{i,j,k}$. Furthermore, if $\tau_{i,j,0}$ is treated such that $r_{s-1}(i,j,0) = 0$, then all tasks $\tau_{i,j,1}$, $\tau_{i,j,2}$, $\dots$ were removed from $\mathcal{T}$. Hence, it holds that $m \coloneqq r_{s-1}(i,j,0) > 0$. By Observation \ref{obs:remain-valid} and the fact that exponential functions grow faster than any polynomial, it is possible to choose some $r_{s-1}$-valid $v \sqsupseteq w_{s-1}$ such that $||v|| = 0^n$ with $n \in H_m$ and $2^{n-2} > 6(p_{i+j}(p_k(n)))^2+1$. 
        
        We demonstrate that there must always be a word $z$ satisfying Case \ref{task3:1-i} or Case \ref{task3:1-ii}. \emph{Assume, for the sake of contradiction,} that neither Case \ref{task3:1-i} nor \ref{task3:1-ii} are possible, i.e.,
        \begin{align}\label{align:z00}
            \text{for all } z \in 00\Sigma^{n-2} \text{ it holds that } N_i^{v \cup \{z\} \cup A}(F_k^{v \cup \{z\} \cup A}(0^n)) \text{ accepts}
        \end{align}
        and
        \begin{align}\label{align:z01}
            \text{for all } z \in 01\Sigma^{n-2} \text{ it holds that } N_j^{v \cup \{z\} \cup A}(F_k^{v \cup \{z\} \cup A}(0^n)) \text{ accepts.}
        \end{align}
        Our goal is to derive a contradiction to the oracle construction by showing that $r_{s-1}(i,j,0) = 0$ instead of $r_{s-1}(i,j,0) > 0$ would have been possible and would have been the preferred choice. We do this by extending $v$ using two words of length $n$ such that both computations $N_i(F_k(0^n))$ and $N_j(F_k(0^n))$ accept. Intuitively, such two words must exist, because both computations accept for at least $2^{n-2}$ pairwise different single word extensions of $v$.  
        
        There are $2^{n-1}$ words of length $n$ starting with $0$ and an accepting path of the computation $N_j(F_k(0^n))$ queries at most $p_j(p_k(n)) + p_k(n) \leq 2p_j(p_k(n))$ of these words. Let $I$ be a set consisting of $2p_j(p_k(n)) + 1$ pairwise different words from $00\Sigma^{n-2}$. By the choice of $n$ it holds that
        \[2^{n-2} > 2p_j(p_k(n))+1,\]
        whereby such a set exists. Recall the definition of $Q(\cdot )$ and consider the following set 
        \[Q \coloneqq 01\Sigma^{n-2} \cap \bigcup _{x \in I} Q(N_i^{v \cup \{x\} \cup A}(F_k^{v \cup \{x\} \cup A}(0^n))).\]
        Since $N_i(F_k(0^n))$ accepts relative to $v \cup \{x\} \cup A$ for all $x \in I$ and the leftmost accepting path queries at most $\leq 2p_i(p_k(n))$ words, we get
        \[\card{Q} \leq (2p_j(p_k(n))+1) \cdot 2p_i(p_k(n)) \leq 6(p_{i+j}(p_k(n)))^2.\]
        Choose some $z_1 \in 01\Sigma^{n-2} \setminus Q$ (exists by the choice of $n$). By $\card{I} = 2p_j(p_k(n))+1$, there is also some  
        \[z_0 \in I \setminus Q(N_j^{v \cup \{z_1\} \cup A}(F_k^{v \cup \{z_1\} \cup A}(0^n))).\]
        Next we show that both $N_i(F_k(0^n))$ and $N_j(F_k(0^n))$ accept relative to $v \cup \{z_0,z_1\} \cup A$. By (\ref{align:z00}), $N_i(F_k(0^n))$ accepts relative to $v \cup \{z_0\} \cup A$. Since $z_0 \in I$, $Q$ contains all queries of the leftmost accepting path of this computation. Furthermore $z_1 \notin Q$, so $N_i(F_k(0^n))$ also accepts relative to $v \cup \{z_0,z_1\} \cup A$. By (\ref{align:z01}), $N_j(F_k(0^n))$ accepts relative to $v \cup \{z_1\} \cup A$. Since $z_0$ is chosen such that it avoids the words queried on the leftmost accepting path of this computation, $N_j(F_k(0^n))$ also accepts relative to $v \cup \{z_0,z_1\} \cup A$. 
        
        Let $\hat{s} < s$ be the step of the oracle construction treating the task $\tau_{i,j,0}$. As discussed earlier, $\tau_{i,j,0}$ was treated by Case \ref{task2:2}. We show that treating $\tau_{i,j,0}$ by Case \ref{task2:1} would have been possible, giving a contradiction to the construction of the oracle, since Case \ref{task2:1} would have been preferred over Case \ref{task2:2}.
        
        By Observation \ref{obs:valid-monotone} and since $v$ is $r_{s-1}$-valid, $v$ is also $r_{\hat{s}-1}$-valid. Also, $v \sqsupseteq w_{s-1} \sqsupseteq w_{\hat{s}-1}$. Let $v'$ be the extension of $v$ such that $v'$ is defined for all words of length $2p_{i+j}(p_k(n))$ and $v' \setminus v = \{z_0,z_1\}$. We show that $v'$ is $r_{\hat{s}-1}$-valid.
        
        V\ref{V1}($v'$) is satisfied, because $z_0$ and $z_1$ are the only words of their length in $v'$. V\ref{V2}($v'$) and V\ref{V3}($v',r_{\hat{s}-1}$) are satisfied, because $z_0$ and $z_1$ have even length and thus are no list encodings $\langle \cdot \rangle$. Furthermore, V\ref{V4}($v',r_{\hat{s}-1}$) is satisfied, because $r_{\hat{s}-1}(i,j,0)$ is not defined and $m \notin \dom(r_{\hat{s}-1})$, hence $z_0,z_1$ of length $n \in H_m$ are not forbidden to be inside $v'$. In total, $v' \sqsupseteq w_{\hat{s}-1}$ is $r_{\hat{s}-1}$-valid and there is an $x \coloneqq F_k^{v' \cup A}(0^n)$ such that both $N_i^{v' \cup A}(x)$ and $N_j^{v' \cup A}(x)$ accept definitely. So the oracle construction would have preferred $w_{\hat{s}} \coloneqq v'$ with $r_{\hat{s}} \coloneqq r_{\hat{s}-1} \cup \{(i,j,0) \mapsto 0\}$ over the actual choice of $w_{\hat{s}}$ and $r_{\hat{s}}$ with $r_{\hat{s}}(i,j,0) > 0$, a contradiction to the construction of $w_s$ and $r_s$.
        
        Consequently, the above assumption is false. So let $z$ be a word satisfying either Case \ref{task3:1-i} or Case \ref{task3:1-ii} of task $\tau_{i,j,k}$. Using analogous arguments as for $v'$ before, V\ref{V1}($v \cup \{z\}$), V\ref{V2}($v \cup \{z\}$) and V\ref{V3}($v \cup \{z\},r_{s-1}$) are satisfied. For V\ref{V4}($v \cup \{z\},r_{s-1}$) note that $v$ does not contain any word of length $|z|$. So V\ref{V4}($v \cup \{z\},r_{s-1}$) is also satisfied. By Observation \ref{obs:remain-valid}, we can extend $v \cup \{z\}$ to $w_s$ via zeros to the required length and remain $r_{s-1}$-valid. Since the extension of $r_{s-1}$ to $r_s$ via $\{(i,j,k) \mapsto 0\}$ does not affect any requirement, $w_s$ is also $r_s$-valid. 
    \end{claimproof}

    Above claims show that $w_s$ and $r_s$ are always well-defined and $w_s$ is $r_s$-valid. Furthermore, all tasks clearly ensure that $w_{s+1} \sqsupsetneq w_{s}$ and $r_{s+1} \sqsupsetneq r_{s}$. Hence, the lemma follows.    
\end{proof}
\begin{lemma}\label{lem:O-rvalid}
    $O$ is $r$-valid.
\end{lemma}
\begin{proof}
    Follows from Lemma \ref{lem:ws-rs-valid} and that any violation of V\ref{V1} to V\ref{V4} would also be a violation for some pair $(w_s,r_s)$ with $s \in \N$.
\end{proof}
\begin{proposition}\label{prop:O-is-sparse}
    $O \in \SPARSE$.
\end{proposition}
\begin{proof}
    Follows by V\ref{V1}($O$) which holds by Lemma \ref{lem:O-rvalid}.
\end{proof}
\begin{proposition}\label{prop:no-jump-operator}
    There are no recursive jump operators for $\TAUT$ relative to $O \cup A$.
\end{proposition}
\begin{proof}
    Assume for contradiction that $J_i^{O \cup A}$ is a recursive jump operator. Let $s$ be the step of the oracle construction treating the task $\tau_{i,0,0}$. We make a case distinction on the definition of $r_s(i,0,0)$ and in all cases provide a contradiction to above assumption.
    \medskip
    \\
    \centerline{\textbf{Case} $r_s(i,0,0) = 0$:}
    In this case there is an $a \in \N^+$ such that $J_i^{w_s \cup A}(a) = b$ is definite (cf.~Requirement~\ref{task1:1-ii}). So $J_i^{O \cup A}(a) = b$. Furthermore, there is an $x$ such that $P_b^{w_s \cup A}(x) \notin \TAUT^{w_s \cup A}$ is definite (cf.~Requirement~\ref{task1:1-iii}). So $P_b^{O \cup A} \notin \TAUT^{O \cup A}$, which means that $P_b^{O \cup A}$ is no proof system for $\TAUT^{O \cup A}$. Moreover, $P_a$ is a proof system for $\TAUT$ relative to $A$ combined with any $r_{s-1}$-valid fully defined extension of $w_s$ (cf.~Requirement \ref{task1:1-i}). Since $O \sqsupseteq w_s$ is $r$-valid (Lemma \ref{lem:O-rvalid}), $r \sqsupseteq r_{s-1}$ and $O$ is fully defined, we get that $P_a^{O \cup A}$ is a proof system for $\TAUT^{O \cup A}$. All in all, we get a contradiction to the assumption that $J_i^{O \cup A}$ is a jump operator, because $J_i^{O \cup A}(a) = b$, $P_a^{O \cup A}$ is a proof system for $\TAUT^{O \cup A}$, but $P_b^{O \cup A}$ is no proof system for $\TAUT^{O \cup A}$.
    \medskip
    \\
    \centerline{\textbf{Case} $r_s(i,0,0) = 1$:}
    In this case there is an $a \in \N^+$ such that $J_i(a)$ never halts relative to $A$ combined with any $r_{s-1}$-valid partial extension of $w_{s-1}$. From this follows that $J_i^{O \cup A}(a)$ also does not halt, because otherwise there would be some $s' \geq s$ such that $J_i^{w_{s'} \cup A}(a)$ halts and $w_{s'} \sqsupseteq w_{s-1}$ is $r_{s-1}$-valid, which follows from $r_{s-1} \sqsubseteq r_{s'}$ and the $r_{s'}$-validity of $w_{s'}$. Since $J_i^{O \cup A}(a)$ does not halt, $J_i^{O \cup A}$ is not recursive, a contradiction to the assumption that $J_i^{O \cup A}$ is a recursive jump operator.
    \medskip
    \\
    \centerline{\textbf{Case} $r_s(i,0,0) > 1$:}
    The remaining case is the most complex one. The goal is to invoke Lemma \ref{lem:Wim-simulate} and show that the assumed jump operator $J_i^{O \cup A}$ maps a proof system for $\TAUT^{O \cup A}$ onto another proof system for $\TAUT^{O \cup A}$, but the second is simulated by the first. Let 
    \[m \coloneqq r_s(i,0,0) > 1.\]
    Consider the witness proof system $W_{i,m}^{O \cup A}$ from Definition \ref{def:witness-function-jump-operator} with code $\fp(\langle s_{i,m} \rangle)$ (cf.~Observation \ref{obs:Wim}). Let 
    \[b \coloneqq J_i^{O \cup A}(\fp(\langle s_{i,m} \rangle)),\] 
    which exists by Step \ref{task1:3-ii}. 
    Step \ref{task1:3-iii} ensures that $\langle 0^i, 0^m, 0^n \rangle \in O$ for some $n \in \N$. 
    By V\ref{V3}($O,r$) and $r(i,0,0) = m > 1$, it holds that $\langle 1^i, 1^m, 1^n \rangle \notin O$ for all $n \in \N$. 
    In the next paragraph, we show that $P_b^{O \cup A}$ is a proof system for $\TAUT^{O \cup A}$. Then we can invoke Lemma \ref{lem:Wim-simulate} and obtain that $W_{i,m}^{O \cup A}$ is a proof system for $\TAUT^{O \cup A}$ simulating $P_b^{O \cup A}$. This gives the following contradiction to $J_i^{O \cup A}$ being a jump operator: $J_i^{O \cup A}(\fp(\langle s_{i,m} \rangle)) = b$, $P_{\fp(\langle s_{i,m} \rangle)}^{O \cup A} = W_{i,m}^{O \cup A}$ is a proof system for $\TAUT^{O \cup A}$, but $W_{i,m}^{O \cup A}$ simulates $P_b^{O \cup A}$.

    It remains to show that $P_b^{O \cup A}$ is a proof system for $\TAUT^{O \cup A}$. \emph{Assume otherwise for the sake of contradiction}, then $P_b^{O \cup A}$ is no proof system for $\TAUT^{O \cup A}$ for one of the following two reasons: $P_b^{O \cup A}$ does not run in polynomial time or $\ran(P_b^{O \cup A}) \neq \TAUT^{O \cup A}$. 

    \begin{claim}\label{claim:Jb-subset}
       If $\ran(P_b^{O \cup A}) \subseteq \TAUT^{O \cup A}$, then this contradicts that $J_i^{O \cup A}$ is a jump operator for $\TAUT^{O \cup A}$.
    \end{claim}
    \begin{claimproof}
    If $\ran(P_b^{O \cup A}) \subseteq \TAUT^{O \cup A}$, then by Lemma \ref{lem:Wim-polytime}, $W_{i,m}^{O \cup A}$ with code $\fp(\langle s_{i,m} \rangle)$ is a proof system for $\TAUT^{O \cup A}$. Consequently, since $P_b^{O \cup A}$ is no proof system for $\TAUT^{O \cup A}$, we get a contradiction that $J_i^{O \cup A}$ is a jump operator for $\TAUT^{O \cup A}$, because $J_i^{O \cup A}(\fp(\langle s_{i,m} \rangle)) = b$ with $P_{\fp(\langle s_{i,m} \rangle)}^{O \cup A} = W_{i,m}^{O \cup A}$ is a proof system for $\TAUT^{O \cup A}$, but $P_b^{O \cup A}$ is not.
    \end{claimproof}

    \begin{claim}\label{claim:Jb-no-subset}
        If $\ran(P_b^{O \cup A}) \not \subseteq \TAUT^{O \cup A}$, then this contradicts the treatment of task $\tau_{i,0,0}$ such that $r_s(i,0,0) > 1$.
    \end{claim}
    \begin{claimproof}
    We show that the task $\tau_{i,0,0}$ would have been treated such that $r_s(i,0,0) = 0$ instead of $r_s(i,0,0) > 1$. By $\ran(P_b^{O \cup A}) \not \subseteq \TAUT^{O \cup A}$ there exists some $x$ such that 
    \[P_b^{O \cup A}(x) \notin \TAUT^{O \cup A}\] 
    and 
    \begin{align}\label{align:minimality-x}
       \forall y \text{ with } P_b^{O \cup A}(y) \notin \TAUT^{O \cup A} \text{ it holds that }\runtime(P_b^{O \cup A}(y)) \geq \runtime(P_b^{O \cup A}(x)).  
    \end{align}
    Let 
    \[n \coloneqq \runtime(P_b^{O \cup A}(x)).\]
    
    \emph{Proof idea:} Our approach is to fix the computations $P_b^{O \cup A}(x) \notin \TAUT^{O \cup A}$ by choosing a partial oracle that is an extension of $w_{s-1}$ where the computations remain unchanged. Simultaneously, we signalize to $W_{i,m}$ to stop simulating $P_b$. We do this by adding $\langle 1^i, 1^m, 1^{p_t(n)} \rangle$ to the oracle. This word is long enough for $P_b(x) \notin \TAUT$ to remain unchanged, but short enough to still prevent $W_{i,m}$ from simulating $P_b(x)$ and any longer computation of $P_b$. The rigorous arguments follow, but they are particularly technical.       
    
    Let $w'$ be chosen such that 
    \[w_{s} \sqsubseteq w' \sqsubseteq O \text{ and } ||w'|| = p_t(n)+1.\] 
    By Observation \ref{obs:valid-monoton2}, $w'$ is $r_{s-1}$-valid. Extend $w'$ to $w$ via zeros such that $|w| = \langle 1^i, 1^m, 1^{p_t(n)} \rangle$. Note that there is at least one full stage extended solely by zeros from $w'$ to $w$. By Observation \ref{obs:remain-valid}, $w$ is also $r_{s-1}$-valid. We show that $w1$ is an oracle that would satisfy all requirements of Case \ref{task1:1} when treating $\tau_{i,0,0}$. Let $v$, $a$, and $x$ of Case \ref{task1:1} be chosen as $w1$, $\fp(\langle s_{i,m} \rangle)$, and $x$ respectively.
    \begin{itemize}
    \item First, we show that $w1 \sqsupsetneq w_{s-1}$ is $r_{s-1}$-valid. Since $\langle 1^i, 1^m, 1^{p_t(n)} \rangle$ is the only word of its length in $w1$, V\ref{V1}($w1$) and V\ref{V4}($w1,r_{s-1}$) are satisfied. Since $r_{s-1}(i,0,0)$ is not defined, also V\ref{V3}($w1,r_{s-1}$) is satisfied. Finally, V\ref{V2}($w1$) is satisfied, because $w_s$ was built by Case \ref{task1:3}. So through Step \ref{task1:3-iii}, it already holds that $\langle 0^i, 0^m, 0^{n'} \rangle \in w_s$ for some $n' \in \N$. Since V\ref{V2}($w_s$) is satisfied and $w \sqsupseteq w_s$, it holds that 
    \[|\langle 1^i, 1^m, 1^{p_t(n)} \rangle| \geq |\langle 0^i, 0^m, 0^{n'} \rangle| > \runtime(J_i^{w_s \cup A}(\fp(\langle s_{i,m} \rangle))) = \runtime(J_i^{w1 \cup A}(\fp(\langle s_{i,m} \rangle))).\]
    \item Requirement \ref{task1:1-ii} holds (i.e., $J_i^{w1 \cup A}(a) = b$ is definite), because by $w_s \sqsubseteq O$ and the argument above 
    \[J_i^{O \cup A}(\fp(\langle s_{i,m} \rangle)) = J_i^{w_s \cup A}(\fp(\langle s_{i,m} \rangle)) \text{ is definite}.\] 
    So $J_i^{w_s \cup A}(\fp(\langle s_{i,m} \rangle)) = b$ is definite. Then from $w_s \sqsubseteq w1$ it also follows that $J_i^{w1 \cup A}(\fp(\langle s_{i,m} \rangle)) = b$ is definite.
    \item For Requirement \ref{task1:1-iii}, observe that $P_b^{O \cup A}(x) \notin \TAUT^{O \cup A}$ can query words of length at most $p_t(n)$. Since $w' \sqsubseteq O$ and $||w'|| = p_t(n) + 1$, it follows that $P_b^{w' \cup A}(x) \notin \TAUT^{w' \cup A}$ is definite. Hence, by $w1 \sqsupseteq w'$, also $P_b^{w1 \cup A}(x) \notin \TAUT^{w1 \cup A}$ is definite.
    \item For Requirement \ref{task1:1-i}, note that we chose $a$ as $\fp(\langle s_{i,m} \rangle)$ such that $P_a^{w1 \cup A} = W_{i,m}^{w1 \cup A}$. Let $v \sqsupseteq w1$ be an arbitrary fully defined $r_{s-1}$-valid extension of $w1$.
    
    Suppose that there is some $\langle \hat{x}, \hat{y} \rangle$ such that $W_{i,m}^{v \cup A}(\langle \hat{x}, \hat{y} \rangle) \notin \TAUT^{v \cup A}$. Then the output can only originate from line \ref{alg1:line-6}, i.e., $P_b^{v \cup A}(\hat{x})$ needs $\leq |\hat{y}|$ steps and outputs an element outside of $\TAUT^{v \cup A}$. Note that the number $b$ computed in line \ref{alg1:line-2} is the same as $b$ defined above. We distinguish two cases.

    If $|\hat{y}| < n$, then $p_t(|\hat{y}|) < p_t(n)$, so $P_b^{v \cup A}(\hat{x}) \notin \TAUT^{v \cup A}$ only depends on words of length $<p_t(n)$. By $w' \sqsubseteq w1 \sqsubseteq v$ and $||w'|| > p_t(n)$, it follows that $P_b^{v \cup A}(\hat{x}) = P_b^{w' \cup A}(\hat{x}) \notin \TAUT^{w' \cup A}$ is definite. Hence, by $w' \sqsubseteq O$, also $P_b^{O \cup A}(\hat{x}) \notin \TAUT^{O \cup A}$. This is a contradiction to the minimality of $x$ (cf.~(\ref{align:minimality-x})), because 
    \[\runtime(P_b^{O \cup A}(\hat{x})) \leq |\hat{y}| < n = \runtime(P_b^{O \cup A}(x)).\] 
    
    If $|\hat{y}| \geq n$, then line \ref{alg1:line-4} of $W_{i,m}^{v \cup A}(\langle \hat{x}, \hat{y} \rangle)$ evaluates to FALSE, because $\langle 1^i, 1^m, 1^{p_t(n)} \rangle \in w1 \sqsubseteq v$ and $p_t(n) \leq p_t(|\hat{y}|)$. This is a contradiction that on input $\langle \hat{x}, \hat{y} \rangle$ line \ref{alg1:line-6} is reached.

    In total, above supposition about the existence of $\langle \hat{x}, \hat{y} \rangle$ is false and thus $\ran(W_{i,m}^{v \cup A}) = \TAUT^{v \cup A}$. Furthermore, by Lemma \ref{lem:Wim-polytime}, $W_{i,m}^{v \cup A}$ runs in polynomial time. So $W_{i,m}^{v \cup A}$ is a proof system for $\TAUT^{v \cup A}$. Since $v \sqsupseteq w1$ is an arbitrary $r_{s-1}$-valid fully defined extension of $w1$, Requirement \ref{task1:1-i} is satisfied. 
    \end{itemize}     
    With this, we have shown that $w1$ would have satisfied all requirements of Case \ref{task1:1} when treating task $\tau_{i,0,0}$ at step $s$ of the oracle construction. So the oracle construction would have constructed $w_s$ and $r_s$ via Case \ref{task1:1} resulting in $r_s(i,0,0) = 0$, a contradiction to the actual choice of $r_s$ with $r_s(i,0,0) > 1$.
    \end{claimproof}
    
    We return to the proof of the third case of Proposition \ref{prop:no-jump-operator}. It either must hold that $P_b^{O \cup A} \subseteq \TAUT^{O \cup A}$ or that $P_b^{O \cup A} \not \subseteq \TAUT^{O \cup A}$. The first case gives a contradiction via Claim \ref{claim:Jb-subset}, the second via Claim \ref{claim:Jb-no-subset}. Hence, the assumption that $P_b^{O \cup A}$ is no proof system for $\TAUT^{O \cup A}$ is false.
\end{proof}

\begin{proposition}\label{prop:no-optimal-ps}
    There are no $\leqmpp$-complete disjoint $\NP$-pairs relative to $O \cup A$.
\end{proposition}
\begin{proof}
    Assume otherwise for the sake of contradiction and let $(C,D)$ be a $\leqmpp$-complete disjoint $\NP$-pair relative to $O \cup A$. Let $i,j \in \N^+$ such that $C = L(N_i^{O \cup A})$ and $D = L(N_j^{O \cup A})$. 

    We first show that $r(i,j,0) > 0$. Suppose that $r(i,j,0) = 0$, then let $s$ be the step treating the task $\tau_{i,j,0}$. Consequently, $r_s(i,j,0) = 0$ and task $\tau_{i,j,0}$ was treated by Case \ref{task2:1}. So there is an $x$ such that $N_i^{w_s \cup A}(x)$ and $N_j^{w_s \cup A}(x)$ both accept definitely. By $O \sqsupseteq w_s$, this behavior is also the same relative to $O \cup A$. But then $L(N_i^{O \cup A}) \cap L(N_j^{O \cup A}) \neq \emptyset$, a contradiction that $(C,D) \in \DisjNP^{O \cup A}$. This shows that $r(i,j,0) > 0$.

    Let $m \coloneqq r(i,j,0)$. Consider the witness pair $(A_m,B_m)$. We show that $(A_m,B_m) \in \DisjNP^{O \cup A}$ and $(A_m,B_m) \not \leqmpp (C,D)$, providing a contradiction to the assumed $\leqmpp$-completeness of $(C,D)$.

    By Lemma \ref{lem:O-rvalid}, $O$ is $r$-valid. Thus, V\ref{V4}($O,r$) is satisfied. Note that $A \subseteq 1\Sigma^*$ and $O \subseteq 0\Sigma^*$. So for each $n \in H_m$ there is at most one word from $0\Sigma^{n-1}$ in $O \cup A$. By Observation \ref{obs:witness-disjoint}, $(A_m^{O \cup A},B_m^{O \cup A}) \in \DisjNP^{O \cup A}$.

    Let $f \in \FP^{O \cup A}$ be the function showing that $(A_m,B_m) \leqmpp (C,D)$. Let $k \in \N^+$ be such that $\pi$ is computed by $F_k^{O \cup A}$. Let $s$ be the step treating task $\tau_{i,j,k}$. Let $n \in H_m$ be as chosen when treating this task.

    If Case \ref{task3:1-i} was achieved, then $0^n \in A_m^{w_s \cup A}$. Since $||w_s|| > n$ and $O \sqsupseteq w_s$, it follows that $0^n \in A_m^{O \cup A}$. Furthermore, $N_i^{w_s \cup A}(F_k^{w_s \cup A}(0^n))$ rejects. Since $||w_s|| > p_i(p_k(n)) + p_k(n)$, this behavior is definite. Hence, $N_i^{O \cup A}(F_k^{O \cup A}(0^n))$ also rejects. But then $0^n$ is a witness showing that $F_k^{O \cup A}(A_m^{O \cup A}) \not \subseteq L(N_i^{O \cup A})$, a contradiction to $(A_m^{O \cup A}, B_m^{O \cup A}) \leqmpp (C,D)$ via $\pi$.

    If Case \ref{task3:1-ii} was achieved, the arguments are symmetric. Here, it holds that $0^n \in B_m^{w_s \cup A}$. Since $||w_s|| > n$ and $O \sqsupseteq w_s$, it follows that $0^n \in B_m^{O \cup A}$. Furthermore, $N_j^{w_s \cup A}(F_k^{w_s \cup A}(0^n))$ rejects. Since $||w_s|| > p_j(p_k(n)) + p_k(n)$, this behavior is definite. Hence, $N_j^{O \cup A}(F_k^{O \cup A}(0^n))$ rejects also. But then $0^n$ is a witness showing that $F_k^{O \cup A}(B_m^{O \cup A}) \not \subseteq L(N_j^{O \cup A})$, a contradiction to $(A_m^{O \cup A}, B_m^{O \cup A}) \leqmpp (C,D)$ via $\pi$.

    In total, the assumption at the start of the proof must be false and there is no $\leqmpp$-complete disjoint $\NP$-pair relative to $O \cup A$.
\end{proof}
\begin{theorem}\label{thm:main-result}
    Let $A \subseteq 2\N$ be arbitrary. There is a sparse oracle $O \subseteq 2\N+1$ relative to which $\DisjNP^{O \cup A}$ has no $\leqmpp$-complete pairs, $\TAUT^{O \cup A}$ has no optimal proof systems, and $\TAUT^{O \cup A}$ has no recursive jump operators relative to $O \cup A$.
\end{theorem}
\begin{proof}
    Take oracle $O$ from Definition \ref{def:desired-oracle}. Then the properties for $O$ and $O \cup A$ follow by Propositions \ref{prop:O-is-sparse}, \ref{prop:no-jump-operator}, \ref{prop:no-optimal-ps}, and the fact that the non-existence of $\leqmpp$-complete pairs for $\DisjNP$ relativizably implies that $\TAUT$ has no optimal proof systems by Razborov \cite{raz94}. 
\end{proof}
\begin{corollary}\label{cor:main-result}
    There is an oracle $B$ relative to which $\PH^B$ is infinite, $\DisjNP^B$ has no $\leqmpp$-complete pairs, $\TAUT^B$ has no optimal proof systems, and $\TAUT^B$ has no recursive jump operators relative to $B$.
\end{corollary}
\begin{proof} 
    Let $A$ be the oracle of Yao \cite{yao85} relative to which the polynomial-time hierarchy is infinite. Without loss of generality, $A \subseteq 2\N$. Let $O \subseteq 2\N+1$ be the sparse oracle from Theorem \ref{thm:main-result} such that the properties of Theorem \ref{thm:main-result} hold relative to $O \cup A$.  
    
    By results of Balcázar, Book, and Schöning \cite{bbs86} and Long and Selman \cite{ls86}, if the polynomial-time hierarchy is infinite, then the polynomial-time hierarchy is also infinite relative to any sparse oracle. Egidy and Glaßer \cite[Cor.~3.18]{eg25} show that this result allows us to combine above oracle $A \subseteq 2\N$ relative to which the polynomial-time hierarchy is infinite with above sparse oracle $O \subseteq 2\N+1$ and the polynomial-time hierarchy remains infinite relative to $O \cup A$. Then $B \coloneqq O \cup A$ is the desired oracle of this theorem.
\end{proof}
\begin{corollary}
    Question \emph{Q1} can not be answered in the positive by relativizable means, even when assuming that the polynomial-time hierarchy is infinite.
\end{corollary}

\bibliography{Literatursammlung}

\end{document}